\documentclass[]{article}
\usepackage[english]{babel}
\usepackage{tikz}
\usepackage{amssymb}
\usepackage[T1]{fontenc}
\usepackage[utf8]{inputenc}
\usepackage{calrsfs}
\usepackage{amsmath}
\usepackage{amsthm}
\usepackage{enumerate}
\usepackage{caption}
\usetikzlibrary[shapes.misc]
\usetikzlibrary[shapes.geometric]
\newtheorem{de}{Definition}[section]
\newtheorem{theo}{Theorem}[section]
\newtheorem{prop}[theo]{Proposition}
\newtheorem{cor}[theo]{Corollary}
\newtheorem{lem}[theo]{Lemma}

\newtheorem{obs}[theo]{Claim}
\newtheorem{con}[theo]{Claim}

\title{Dichotomies properties on computational complexity of $S$-packing coloring problems}
\author{Nicolas Gastineau}
\begin{document}
\maketitle
\begin{abstract}
This work establishes the complexity class of several instances of the $S$-packing coloring problem: for a graph $G$, a positive integer $k$ and a nondecreasing list of integers $S=(s_1,\ldots, s_{k})$, $G$ is \emph{$S$-colorable} if its vertices can be partitioned into sets $S_{i}$, $i=1,\ldots, k$, where each $S_{i}$ is an $s_i$-packing (a set of vertices at pairwise distance greater than $s_i$).
In particular we prove a dichotomy between NP-complete problems and polynomial-time solvable problems for lists of at most four integers.
\end{abstract}
\section{Introduction}
We consider only finite undirected connected graphs.
For a graph $G$, an \emph{$i$-packing} is a set $X_{i}\subseteq V(G)$ such that for any distinct pair $u$, $v$ $\in$ $X_{i}$, $d_{G}(u,v)> i$, where $d_{G}(u,v)$ denotes the distance between $u$ and $v$. We will use $X_{i}$ to refer to an $i$-packing in a graph $G$.
For a nondecreasing sequence of positive integers $S=\{s_i|i>0 \}$, a \emph{$S$-$k$-(packing)-coloring} of $G$ is a partition of $V(G)$ into sets $S_1,\ldots ,S_k$, where each $S_{i}$ is an $s_i$-packing.
For a nondecreasing list of $k$ integers $S'=(s_{1},\ldots,s_{k})$, an \emph{$S'$-(packing) coloring} of $G$ is an $S$-$k$-coloring of $G$ for a sequence $S$ which begins with a list $S'$.
A graph $G$ is \emph{$S$-colorable} if there exists an $S$-coloring of $G$.

The \textbf{S-COL} decision problem consists in determining, for fixed $S$, if $G$ is $S$-colorable, for a graph $G$ as input.

By $|S|$, we denote the size of a list $S$, and by $s_i$ we denote the $i$th element of a list.
Let $S=(s_1,s_2,\ldots)$ and $S'=(s'_1,s'_2,\ldots)$ be two nondecreasing lists of integers with $|S|=|S'|$. We define an order on the lists by $S\le S'$ if $s_i\ge s'_i$, for every integer $i$, $1\le i\le |S|$. Note that if $S\le S'$, $G$ is $S$-colorable implies $G$ is $S'$-colorable.

In this article, for a $(s_{1},s_{2},\ldots)$-coloring of a graph, we prefer to map vertices to the color multi-set $\{s_{1},s_{2},\ldots\}$ even if two colors can be denoted by the same number. This notation allows the reader to directly see to which type of packing the vertex belongs depending on its color. When needed, we will denote colors of vertices in different $i$-packings by $i_a,i_b,\ldots$.

Let $S_{d}^{k}$ be a list only containing $k$ integers $d$.
The problem $S_{1}^{k}$-COL corresponds to the $k$-coloring problem which is known to be NP-complete for $k\ge 3$.
The $S$-coloring generalizes coloring with distance constraints like the \emph{packing coloring} or the \emph{distance coloring} of a graph.
We denote by P-COL, the problem $(1,2,\ldots,k)$-COL for a graph $G$ and an integer $k$ (with $G$ and $k$ as input).
The packing chromatic number \cite{GO2008} of $G$ is the least integer $k$ such that $G$ is $(1,2,\ldots,k)$-colorable.
A series of works \cite{BR2007,Ek2010,FIN2010,GO2008} considered the packing chromatic number of infinite grids.
The $d$-distance chromatic number \cite{ASH2012} of $G$ is the least integer $k$ such that $G$ is $S_{d}^k$-colorable.
Initially, the concept of $S$-coloring has been introduced by Goddard et al. \cite{GO2008} and Fiala et al. \cite{FI2010}.
The $S$-coloring problem was considered in other papers \cite{GO2012,GOH2012}.

The $S$-coloring problem, with $|S|=3$ has been introduced by Goddard et al. \cite{GO2008} in order to determine the complexity of the packing chromatic number when $k=4$.
Moreover, Goddard and Xu \cite{GO2012} have proven that for $|S|=3$, $S$-COL is NP-complete if $s_1=s_2=1$ or if $s_1=1$ and $s_2=s_3=2$ and polynomial-time solvable otherwise.
About the complexity of $S$-COL, Fiala et al. \cite{FI2010} have proven that P-COL is NP-complete for trees and Argiroffo et al. \cite{ARG2011,AR2012} have proven that P-COL is polynomial-time solvable on some classes of graphs.

In the second section, for a list $S$ of three integers, we determine the family of $S$-colorable trees. 
Moreover, we determine dichotomies on cubic graphs, subcubic graphs and bipartite graphs.
In the third section, we determine polynomial-time solvable and NP-complete instances of $S$-COL, for unfixed size of lists. We use these results to determine a dichotomy between NP-complete instances and polynomial-time solvable instances of $S$-COL for $|S|\le 4$.

Note that for any nondecreasing list of integers $S$, we have $S$-COL in NP.
\section{Complexity of $S$-COL for a list of three integers and several classes of graphs}
This section is dedicated to the proofs of two theorems: Theorem \ref{theo21} and Theorem \ref{theo22}.
In \cite{GO2012}, the family of $S$-colorable graphs, for $|S|=3$, is described in the case $S$-COL is polynomial-time solvable. 
Using the properties of these families of graphs \cite{GO2012}, it is easy to determine the $S$-colorable trees for $|S|=3$ and $S\neq (1,2,2)$. In Theorem \ref{theo21}, we determine the $(1,2,2)$-colorable trees by giving a characterization by forbidding subtrees. The following definition gives a construction of this family of forbidden subtrees.
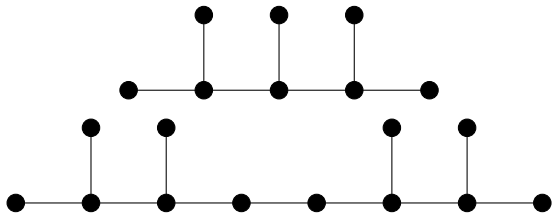
\begin{figure}[t]
\begin{center}
\begin{tikzpicture}
\draw (0+2.5,0) -- (4+2.5,0);
\draw (1+2.5,0) -- (1+2.5,1);
\draw (2+2.5,0) -- (2+2.5,1);
\draw (3+2.5,0) -- (3+2.5,1);
\node at (0+2.5,0) [circle,draw=black,fill=black, scale=0.7] {};
\node at (1+2.5,0) [circle,draw=black,fill=black, scale=0.7] {};
\node at (2+2.5,0) [circle,draw=black,fill=black, scale=0.7] {};
\node at (3+2.5,0) [circle,draw=black,fill=black, scale=0.7] {};
\node at (4+2.5,0) [circle,draw=black,fill=black, scale=0.7] {};
\node at (1+2.5,1) [circle,draw=black,fill=black, scale=0.7] {};
\node at (2+2.5,1) [circle,draw=black,fill=black, scale=0.7] {};
\node at (3+2.5,1) [circle,draw=black,fill=black, scale=0.7] {};

\draw (1,-1.5) -- (8,-1.5);
\draw (2,-1.5) -- (2,-0.5);
\draw (3,-1.5) -- (3,-0.5);
\draw (6,-1.5) -- (6,-0.5);
\draw (7,-1.5) -- (7,-0.5);
\node at (1,-1.5) [circle,draw=black,fill=black, scale=0.7] {};
\node at (2,-1.5) [circle,draw=black,fill=black, scale=0.7] {};
\node at (3,-1.5) [circle,draw=black,fill=black, scale=0.7] {};
\node at (4,-1.5) [circle,draw=black,fill=black, scale=0.7] {};
\node at (5,-1.5) [circle,draw=black,fill=black, scale=0.7] {};
\node at (6,-1.5) [circle,draw=black,fill=black, scale=0.7] {};
\node at (7,-1.5) [circle,draw=black,fill=black, scale=0.7] {};
\node at (8,-1.5) [circle,draw=black,fill=black, scale=0.7] {};
\node at (2,-0.5) [circle,draw=black,fill=black, scale=0.7] {};
\node at (3,-0.5) [circle,draw=black,fill=black, scale=0.7] {};
\node at (6,-0.5) [circle,draw=black,fill=black, scale=0.7] {};
\node at (7,-0.5) [circle,draw=black,fill=black, scale=0.7] {};

\end{tikzpicture}
\end{center}
\caption{The trees $T_0$ (on the top) and $T_1$ (on the bottom).}
\label{forbtree}
\end{figure}

\begin{de}
Let $T_0$ and $T_1$ be the trees from Figure \ref{forbtree}.
We define a family of trees $\mathcal{T}$ as follows:
\begin{itemize}
\item[i)] $T_0\in \mathcal{T},T_1\in \mathcal{T}$;
\item[ii)] if $T\in \mathcal{T}$, then the tree $T'$, obtained from $T$ by removing an edge between two vertices $u$ and $v$ of degree $2$, by adding four vertices $u_1$, $u_2$, $u_3$ and $w$ and by adding the edges $uu_1$, $u_1 u_2$, $u_2 u_3$, $u_3 v$ and $u_2 w$, is in $\mathcal{T}$.
\end{itemize}
\end{de}
We can note that the tree $T_0$ does not contain two adjacent vertices of degree 2. We will prove:
\begin{theo}\label{theo21}
A tree $T$ is $(1,2,2)$-colorable if and only if it does not contain a tree from $\mathcal{T}$.
\end{theo}
The following theorem by Goddard and Xu \cite{GO2012} establishes a dichotomy between NP-complete problems and polynomial-time solvable problems for $|S|=3$.
\begin{theo}[\cite{GO2012}]
Let $k$ be a positive integer. The problems $(1,1,k)$-COL and $(1,2,2)$-COL are both NP-complete. Except these problems, $S$-COL is polynomial-time solvable for $|S|=3$.
\end{theo}
For $(1,1,k)$-COL, the authors of \cite{GO2012} provided a reduction from $3$-COL with a graph of maximal degree $\Delta$, the produced graph has maximal degree $2\Delta$. Since $3$-COL is polynomial-time solvable for small $\Delta$, the proof can not be easily changed to have a reduction with subcubic graphs.
For $(1,2,2)$-COL, the authors provided a reduction from NAE SAT, the produced graph has maximal degree $\Delta$, where $\Delta$ is the maximal number of times a variable can appear positively or negatively. This reduction can not be easily changed to have a reduction with subcubic graphs.
In Theorem \ref{theo22}, we establish similar results for subcubic graphs, cubic graphs and bipartite graphs.

For different instances of $S$-COL, with $|S|=3$, Table \ref{tab1} summarizes the class of complexity of $S$-COL for different classes of graphs.
We recall that every bipartite graph is $(1,1,k)$-colorable, for $k$ a positive integer and that every subcubic graph except $K_4$ is $(1,1,1)$-colorable by Brooks' theorem. We will prove:
\begin{theo}\label{theo22}
Every instance of $S$-COL, for $|S|=3$, is either polynomial-time solvable or NP-complete, for subcubic graphs, cubic graphs and bipartite graphs.
\end{theo}
\begin{table}[!ht]
\centering\begin{tabular}{|c|c|c|c|}
\hline
Class of graphs & $(1,1,1)$-COL & $(1,2,2)$-COL &$(1,1,k)$-COL, ($k\ge2$) \\\hline
Arbitrary graphs& NP-complete & NP-complete & NP-complete \\
Subcubic graphs & Polynomial & NP-complete & NP-complete \\
Cubic graphs & Polynomial & Polynomial & NP-complete \\
Bipartite graphs &  Polynomial & NP-complete &  Polynomial \\
Trees &  Polynomial &  Polynomial &  Polynomial \\ \hline
\end{tabular}
\caption{\label{tab1}Complexity class of several instances of $S$-COL.}
\end{table}
\subsection{Proof of Theorem \ref{theo21}}
By a 3-vertex we denote a vertex of degree at least 3.
\begin{lem}
The trees from $\mathcal{T}$ are not $(1,2,2)$-colorable.
\label{rema}
\end{lem}
\begin{proof}
First, note that we can not color a 3-vertex by $1$.
Thus, the tree $T_0$, containing three 3-vertices at mutual distance at most 2, is not (1,2,2)-colorable.
Second, in the tree $T_1$, one of the two vertices of degree 2 should has some color $2$.
Hence, one vertex among the three 3-vertices at distance at most $2$ from this vertex is not colorable.
For the other trees, we can note that for every pair of adjacent vertices of degree 2, one of the vertex has some color $2$.
Moreover, if a vertex is at distance at most 2 from three vertices, then it can not have a color $2$.
Using these two facts combined together, we can note that if we begin to color two adjacent vertices of degree 3 we cannot extend the $(1,2,2)$-coloring to the other two adjacent vertices of degree 3. 

\end{proof}
\begin{proof}[Proof of Theorem \ref{theo21}]
By Lemma \ref{rema}, if $T$ contains a tree from $\mathcal{T}$, then $T$ is not (1,2,2)-colorable.
If $T$ does not contain $T_0$, then a vertex of degree at least 3 can not be adjacent to two vertices of degree at least 3. Two 3-vertices are \emph{path-connected} if there are only vertices of degree 2 in the path between those two vertices.
Let a 3a-vertex (a 3b-vertex, respectively) be a 3-vertex adjacent (not adjacent, respectively) to another 3-vertex.

We construct a $(1,2,2)$-coloring of $T$ as follows:
take a 3-vertex $u$, give an arbitrary color to this vertex and its possible 3-vertex neighbor and extend the coloring to path-connected 3-vertices. Do the same process from these new colored vertices.
Let $v_1$ and $v_2$ be path-connected 3-vertices at distance $i$. Suppose, without loss of generality, that $v_1$ is colored by $2_a$.
If $i$ is even, then the coloring of $v_1$ can be extended to $v_2$ using the pattern $1,2_b,1,2_a,\ldots,1$ for the vertices in the path between $v_1$ and $v_2$. We give the color $2_a$ (or $2_b$, depending on $i$) to $v_2$.

If $i>3$ is odd then the coloring of $v_1$ can be extended to $v_2$ using the pattern $1,2_b,2_a,1,2_b,1,2_a,\ldots,1$ for the vertices in the path between $v_1$ and $v_2$. We give the color $2_a$ (or $2_b$, depending on $i$) to $v_2$.

If $i=3$, since $T$ does not contain $T_1$, then $v_1$ and $v_2$ are not both 3a-vertices. First, suppose $v_2$ is a 3a-vertex. For the two vertices of the path, we give the color 1 to the vertex adjacent to $v_1$ and the color $2_b$ to the vertex adjacent to $v_2$. Finally, we give the color $2_a$ to $v_2$.
Second, suppose $v_1$ is a 3a-vertex. For the two vertices of the path, we give the color $2_b$ to the vertex adjacent to $v_1$ and the color 1 to the vertex adjacent to $v_2$. Finally, we give the color $2_a$ to $v_2$. Since $G$ does not contain trees from $\mathcal{T}$, $v_1$ is not at distance $2$ from another 3b-vertex and $v_1$ cannot have another neighbor already colored by $2_b$.
\end{proof}
\subsection{Proof of Theorem \ref{theo22}}
\begin{de}
Let $S$ be a nondecreasing list of three integers. Let $s_i$ and $s_j$ be two integers in this list.

An $s_i$-$s_j$-transmitter (respectively, an $s_i$-$s_j$-antitransmitter) is a graph $L$ such that there exist two vertices $u$ and $v$ called connectors and in every $S$-coloring of $L$, $u$ and $v$ have a same color among $\{s_i, s_j\}$ (respectively, different colors among $\{s_i, s_j\}$). Moreover, there should exist an $S$-coloring of $L$ such that every vertex that have a color $c$ different from the color of $u$ (respectively, of $v$) must be at distance greater than $\lfloor c/2 \rfloor$ from $u$ (respectively, from $v$).

An $s_i$-$s_j$-clause simulator is a graph $H$ such that there exist three vertices $u_1$, $u_2$ and $u_3$ called connectors, and in every $S$-coloring of $H$, $u_1$, $u_2$ and $u_3$ have colors among $\{s_i, s_j\}$ and never $u_1$, $u_2$ and $u_3$ have the same color. Moreover, there should exist $S$-colorings of $H$ such that every vertex that have a color $c$ different from the color of $u_1$ (respectively, of $u_2$ , respectively, of $u_3$) must be at distance greater than $\lfloor c/2 \rfloor$ from $u_1$ (respectively, from $u_2$, respectively, from $u_3$), for any possible configuration of colors among $\{s_i,s_j\}$ for the connectors.
\end{de}
\begin{lem}
Let $S$ be a nondecreasing list of three integers and let $s_i$ and $s_j$ be two integers in this list. If there exist an $s_i$-$s_j$-transmitter $L$, an $s_i$-$s_j$-antitransmitter $J$ and an $s_i$-$s_j$-clause simulator $H$, then $S$-COL is NP-complete.
\label{reffinal}
\end{lem}
\begin{proof}
Suppose there exist a transmitter $L$, an antitransmitter $J$ and a clause simulator $H$.
Connecting two vertices $u$ and $v$ by $L$ (respectively, by $J$) corresponds to adding a copy of $L$ (respectively, of $J$) and identifying $u$ with a connector of $L$ (respectively, of $J$) and identifying $v$ with the other connector. The original vertices correspond to the vertices $u$ and $v$.
Let a $L$-$k$-cycle be a cycle of order $k$ in which each edge is replaced by $L$ (it means that if $uv$ is an edge, we remove the edge and connect $u$ and $v$ by $L$).

The proof is by reduction from NAE 3SAT:
a literal is either a variable $x$ (a positive occurrence of $x$) or the negation of a variable (a negative occurrence of $x$).
A clause is a disjunction of literals and the size of a clause is its number of literals.
A truth assignment gives values (true or false) for each variable.
For a given truth assignment on a set of clauses, if there is a true literal and a false literal in each clause, then it satisfies not equally the set of clauses.
\begin{center}
\parbox{10cm}{
\setlength{\parskip}{.05cm}
\textbf{Not-all-equal 3-satisfiability (NAE 3SAT)}

\textbf{Instance}: Collection $C=[c_{1},c_{2},\ldots,c_{m}]$ of clauses on a finite set of variables of size $n$ such that $|c_{i}|=3$ for $1\le i\le\ m$. 

\textbf{Question}: Is there a truth assignment for variables that satisfies not equally all clauses in C ?}
\end{center}

Let $C$ be a collection of clauses.
Let $m$ be the number of clauses and let $n$ be the number of variables.
We construct a graph $G$ as follows.

For each clause, we associate a clause simulator $H$ with connectors $\ell_{1}$, $\ell_{2}$ and $\ell_{3}$ representing the three literals of the clause.
For each variable $x$ with $p_1$ positive occurrences of $x$ and $p_2$ negative occurrences of $x$, we associate a $L-(p_1+1)$-cycle and a $L-(p_2+1)$-cycle where a connector $b_1$ of the $L-(p_1+1)$-cycle and a connector $b_2$ of the $L-(p_2+1)$-cycle are connected by $J$.
For each literal $\ell'_{1}$ which is a positive occurrence of $x$ (respectively, negative occurrence of $x$), we connect $\ell'_{1}$ by $L$ to a connector of the $L-(p_1+1)$-cycle different from $b_1$ (respectively, to a connector of the $L-(p_2+1)$-cycle different from $b_2$), so that every connector of the $L-(p_1+1)$-cycle (respectively, $L-(p_2+1)$) is connected to only one occurrence of $x$.

Suppose there exists a truth assignment of variables that satisfies not equally all clauses in $C$. To each clause $c$ corresponds a clause simulator $H$. Given a truth assignment, we give the color $s_i$ to the connectors representing a true literal and $s_j$ to the connectors representing a false literal. If a connector of a transmitter or an antitransmitter is colored, then we give the other possible color to the other connector.
Using the existing $S$-coloring of $H$, $L$ and $J$, it gives us an $S$-coloring of $G$ (the vertices colored have enough distance between them by hypothesis).

Reciprocally, suppose there exists an $S$-coloring of $G$. For each $H$, the connectors of $H$ have colors among $\{s_i,s_j\}$, by hypothesis. By construction, positive occurrences and negative occurrences of a same variable must have different colors. For each positive or negative occurrence of a variable $x$, we create a truth assignment by giving the value true (respectively, false) to positive occurrences colored by $s_i$ (respectively, $s_j$) and by giving the value false (respectively, true) to negative occurrences colored by $s_i$ (respectively, $s_j$). By construction, this truth assignment satisfies not equally all clauses in $C$.
\end{proof}
Note that if $H$, $L$ and $J$ are subcubic graphs, the connectors of $H$ have degree 2 and if the connectors of $L$ and $J$ have degree 1, then the constructed graph is subcubic.
Moreover, if $H$, $L$ and $J$ are bipartite and any pair of connectors of $L$, $J$ and $H$ are connected by a path of even length, then the constructed graph is bipartite.
\begin{theo}\label{theo23}
The problem $(1,2,2)$-COL is NP-complete for subcubic bipartite graphs.
\end{theo}
\begin{figure}[t]
\begin{center}
\begin{tikzpicture}
\draw (6,0) -- (12,0);
\draw (8,0) -- (8,0.5);
\draw (10,0) -- (10,0.5);
\node at (6,0) [circle,draw=black,fill=black, scale=0.7] {};
\node at (7,0) [circle,draw=black,fill=black, scale=0.7] {};
\node at (8,0) [circle,draw=black,fill=black, scale=0.7] {};
\node at (9,0) [circle,draw=black,fill=black, scale=0.7] {};
\node at (10,0) [circle,draw=black,fill=black, scale=0.7] {};
\node at (11,0) [circle,draw=black,fill=black, scale=0.7] {};
\node at (12,0) [circle,draw=black,fill=black, scale=0.7] {};
\node at (8,0.5) [circle,draw=black,fill=black, scale=0.7] {};
\node at (10,0.5) [circle,draw=black,fill=black, scale=0.7] {};
\node at (6,0.4){$\ell'_{1}$};
\node at (12,0.4){$\ell'_{2}$};
\draw (0,0) -- (4,0);
\draw (2,0) -- (2,0.5);
\node at (0,0) [circle,draw=black,fill=black, scale=0.7] {};
\node at (1,0) [circle,draw=black,fill=black, scale=0.7] {};
\node at (2,0) [circle,draw=black,fill=black, scale=0.7] {};
\node at (3,0) [circle,draw=black,fill=black, scale=0.7] {};
\node at (4,0) [circle,draw=black,fill=black, scale=0.7] {};
\node at (2,0.5) [circle,draw=black,fill=black, scale=0.7] {};
\node at (0,0.4){$\ell'_{1}$};
\node at (4,0.4){$\ell'_{2}$};
\end{tikzpicture}
\end{center}
\caption{The graph $L$ (on the left) and the graph $J$ (on the right).}
\label{gJK}
\end{figure}
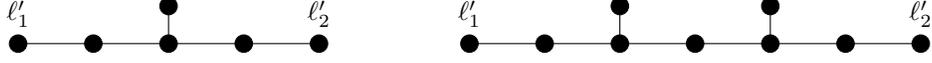

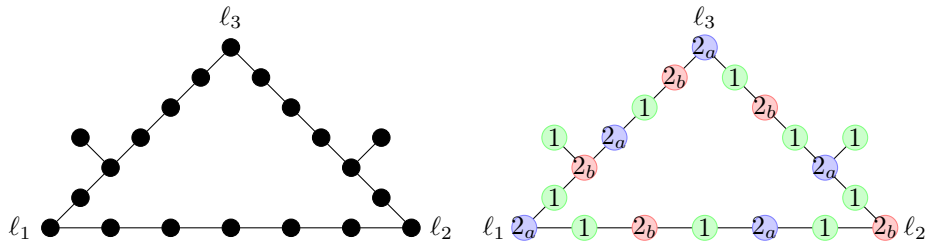
\begin{figure}[t]
\begin{center}
\begin{tikzpicture}
\draw (0,0) -- (4.8,0);
\draw (0,0) -- (2.4,2.4);
\draw (4.8,0) -- (2.4,2.4);
\draw (0.8,0.8) -- (0.4,1.2);
\draw (4,0.8) -- (4.4,1.2);
\node at (0,0) [circle,draw=black,fill=black, scale=0.7] {};
\node at (0.8,0) [circle,draw=black,fill=black, scale=0.7] {};
\node at (1.6,0) [circle,draw=black,fill=black, scale=0.7] {};
\node at (2.4,0) [circle,draw=black,fill=black, scale=0.7] {};
\node at (3.2,0) [circle,draw=black,fill=black, scale=0.7] {};
\node at (4,0) [circle,draw=black,fill=black, scale=0.7] {};
\node at (4.8,0) [circle,draw=black,fill=black, scale=0.7] {};
\node at (0.4,0.4) [circle,draw=black,fill=black, scale=0.7] {};
\node at (0.8,0.8) [circle,draw=black,fill=black, scale=0.7] {};
\node at (1.2,1.2) [circle,draw=black,fill=black, scale=0.7] {};
\node at (1.6,1.6) [circle,draw=black,fill=black, scale=0.7] {};
\node at (2,2) [circle,draw=black,fill=black, scale=0.7] {};
\node at (2.4,2.4) [circle,draw=black,fill=black, scale=0.7] {};
\node at (4.4,0.4) [circle,draw=black,fill=black, scale=0.7] {};
\node at (4,0.8) [circle,draw=black,fill=black, scale=0.7] {};
\node at (3.6,1.2) [circle,draw=black,fill=black, scale=0.7] {};
\node at (3.2,1.6) [circle,draw=black,fill=black, scale=0.7] {};
\node at (2.8,2) [circle,draw=black,fill=black, scale=0.7] {};
\node at (0.4,1.2) [circle,draw=black,fill=black, scale=0.7] {};
\node at (4.4,1.2) [circle,draw=black,fill=black, scale=0.7] {};
\node at (-0.4,0) {$\ell_{1}$};
\node at (5.2,0) {$\ell_{2}$};
\node at (2.4,2.8) {$\ell_{3}$};

\draw (0+6.3,0) -- (4.8+6.3,0);
\draw (0+6.3,0) -- (2.4+6.3,2.4);
\draw (4.8+6.3,0) -- (2.4+6.3,2.4);
\draw (0.8+6.3,0.8) -- (0.4+6.3,1.2);
\draw (4+6.3,0.8) -- (4.4+6.3,1.2);
\node at (0+6.3,0) [circle,draw=blue!50,fill=blue!20] {};
\node at (0+6.3,0) {$2_a$};
\node at (0.8+6.3,0) [circle,draw=green!50,fill=green!20] {};
\node at (0.8+6.3,0) {1};
\node at (1.6+6.3,0) [circle,draw=red!50,fill=red!20] {};
\node at (1.6+6.3,0) {$2_b$};
\node at (2.4+6.3,0) [circle,draw=green!50,fill=green!20] {};
\node at (2.4+6.3,0) {1};
\node at (3.2+6.3,0) [circle,draw=blue!50,fill=blue!20] {};
\node at (3.2+6.3,0) {$2_a$};
\node at (4+6.3,0) [circle,draw=green!50,fill=green!20] {};
\node at (4+6.3,0) {1};
\node at (4.8+6.3,0) [circle,draw=red!50,fill=red!20] {};
\node at (4.8+6.3,0) {$2_b$};
\node at (0.4+6.3,0.4)[circle,draw=green!50,fill=green!20] {};
\node at (0.4+6.3,0.4){1};
\node at (0.8+6.3,0.8) [circle,draw=red!50,fill=red!20] {};
\node at (0.8+6.3,0.8) {$2_b$};
\node at (3.6+6.3,1.2) [circle,draw=green!50,fill=green!20] {};
\node at (3.6+6.3,1.2){1};
\node at (3.2+6.3,1.6) [circle,draw=red!50,fill=red!20] {};
\node at (3.2+6.3,1.6) {$2_b$};
\node at (2.8+6.3,2) [circle,draw=green!50,fill=green!20] {};
\node at (2.8+6.3,2) {1};
\node at (2.4+6.3,2.4) [circle,draw=blue!50,fill=blue!20] {};
\node at (2.4+6.3,2.4){$2_a$};
\node at (4.4+6.3,0.4) [circle,draw=green!50,fill=green!20] {};
\node at (4.4+6.3,0.4){1};
\node at (4+6.3,0.8) [circle,draw=blue!50,fill=blue!20] {};
\node at (4+6.3,0.8){$2_a$};
\node at (1.2+6.3,1.2) [circle,draw=blue!50,fill=blue!20] {};
\node at (1.2+6.3,1.2) {$2_a$};
\node at (1.6+6.3,1.6) [circle,draw=green!50,fill=green!20] {};
\node at (1.6+6.3,1.6) {1};
\node at (2+6.3,2) [circle,draw=red!50,fill=red!20] {};
\node at (2+6.3,2){$2_b$};
\node at (0.4+6.3,1.2) [circle,draw=green!50,fill=green!20] {};
\node at (0.4+6.3,1.2) {1};
\node at (4.4+6.3,1.2) [circle,draw=green!50,fill=green!20] {};
\node at (4.4+6.3,1.2) {1};
\node at (-0.4+6.3,0) {$\ell_{1}$};
\node at (5.2+6.3,0) {$\ell_{2}$};
\node at (2.4+6.3,2.8) {$\ell_{3}$};
\end{tikzpicture}
\end{center}
\caption{The graph $H$ (on the left) and a $(1,2,2)$-coloring of $H$ (on the right).}
\label{gH}
\end{figure}
\begin{proof}
Let $L$ and $J$ be the subcubic graphs from Figure~\ref{gJK} with the corresponding connectors $\ell'_1$ and $\ell'_2$.
Remark that if a vertex has degree 3, then it can not be colored by 1, since the coloring can not be extended to a $(1,2,2)$-coloring.
Therefore, every connector must be colored by $2_a$ or by $2_b$, since in the final constructed graph they will have degree 3.
Observe that there are vertices of degree 3 in $L$ and $J$. Hence, these vertices must also be colored by $2_a$ or by $2_b$.

If the two connectors of $L$ have different colors, then the central vertex of degree 3 can not be colored.
Thus, in every $(1,2,2)$-coloring of $L$, the connectors must have the same color.
A $(1,2,2)$-coloring of $L$ exists by giving the color $2_a$ to the connectors, $2_b$ to the central vertex of degree 3 and 1 to the remaining vertices.

If the two connectors of $J$ have the same color, then the two central vertices of degree 3 can not be colored.
Thus, in every $(1,2,2)$-coloring of $J$, the connectors must be colored differently.
A $(1,2,2)$-coloring of $J$ exists by giving the color $2_a$ to a connector and to a central vertex of degree 3, $2_b$ to the other connector and to the other central vertex of degree 3 and 1 to the remaining vertices.
To conclude, $L$ is a 2-2-transmitter and $J$ is a 2-2-antitransmitter, respectively.\newline

Let $H$ be the subcubic graph from the left part of Figure~\ref{gH} with the corresponding connectors $\ell_1$, $\ell_2$ and $\ell_3$. The graph $H$ contains two vertices of degree 3 and three connectors. Moreover, these vertices should be colored by $2_a$ or by $2_b$.
Suppose $\ell_1$, $\ell_2$ and $\ell_3$ have the same color $2_a$. The two vertices of degree 3 should be colored by $2_b$.
Observe that the vertices at distance 2 from $\ell_3$ could only have the color 1. Thus, the two neighbors of $\ell_3$ can not be colored, since they can not be colored by 1 or by $2_a$ and they can not have both the color $2_b$.
Therefore, in every $(1,2,2)$-coloring of $H$ the connectors must be colored differently.
The right part of Figure~\ref{gH} gives the existence of an $S$-coloring with the suited property for $\ell_1$ and $\ell_3$ with the same color (for the other cases, the (1,2,2)-coloring is similar).
To conclude, $H$ is a 2-2-clause simulator.
\end{proof}
Note that no $r$-regular graph, for $r\ge3$, is $(1,2,2)$-colorable. Thus, this problem seems to be hard only on graphs containing a large proportion of vertices of degree at most 2.

\begin{theo}\label{theo24}
Let $k\ge2$ be a positive integer.
The problem $(1,1,k)$-COL is NP-complete for subcubic graphs.
\label{reffinal2}
\end{theo}
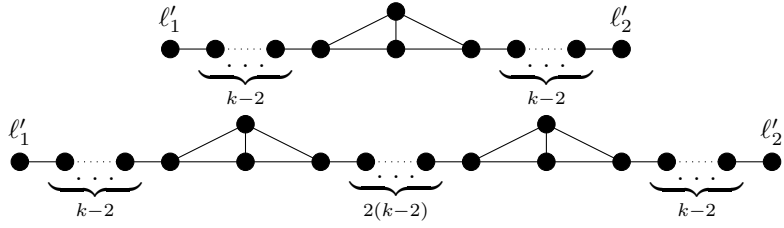
\begin{figure}[t]
\begin{center}
\begin{tikzpicture}
\draw (0,0+0.5) -- (0.6,0+0.5);
\draw[dotted] (0.6,0+0.5) -- (1.4,0+0.5);
\draw[dotted] (4.6,0+0.5) -- (5.4,0+0.5);
\draw (1.4,0+0.5) -- (4.6,0+0.5);
\draw (5.4,0+0.5) -- (6,0+0.5);
\draw (3,0+0.5) -- (3,0.5+0.5);
\draw (2,0+0.5) -- (3,0.5+0.5);
\draw (4,0+0.5) -- (3,0.5+0.5);
\node at (0,0+0.5) [circle,draw=black,fill=black,scale=0.7] {};
\node at (0.6,0+0.5) [circle,draw=black,fill=black,scale=0.7] {};
\node at (1.4,0+0.5) [circle,draw=black,fill=black,scale=0.7] {};
\node at (2,0+0.5) [circle,draw=black,fill=black,scale=0.7] {};
\node at (3,0+0.5) [circle,draw=black,fill=black,scale=0.7] {};
\node at (4,0+0.5) [circle,draw=black,fill=black,scale=0.7] {};
\node at (4.6,0+0.5) [circle,draw=black,fill=black,scale=0.7] {};
\node at (5.4,0+0.5) [circle,draw=black,fill=black,scale=0.7] {};
\node at (6,0+0.5) [circle,draw=black,fill=black,scale=0.7] {};
\node at (3,0.5+0.5) [circle,draw=black,fill=black,scale=0.7] {};
\node at (1,0){$\underbrace{\ \ \ .\  .\  .\ \ \ }_{k-2}$};
\node at (5,0){$\underbrace{\ \ \ .\  .\  .\ \ \ }_{k-2}$};
\node at (0,0.4+0.5){$\ell'_{1}$};
\node at (6,0.4+0.5){$\ell'_{2}$};

\draw (-2,-1) -- (-1.4,-1);
\draw (-0.6,-1) -- (2.6,-1);
\draw (3.4,-1) -- (6.6,-1);
\draw (7.4,-1) -- (8,-1);
\draw[dotted] (-1.4,-1) -- (-0.6,-1);
\draw[dotted] (2.6,-1) -- (3.4,-1);
\draw[dotted] (6.6,-1) -- (7.4,-1);
\draw (1,-1) -- (1,-0.5);
\draw (0,-1) -- (1,-0.5);
\draw (2,-1) -- (1,-0.5);
\draw (5,-1) -- (5,-0.5);
\draw (4,-1) -- (5,-0.5);
\draw (6,-1) -- (5,-0.5);
\node at (-2,-1) [circle,draw=black,fill=black,scale=0.7] {};
\node at (-1.4,-1) [circle,draw=black,fill=black,scale=0.7] {};
\node at (-0.6,-1) [circle,draw=black,fill=black,scale=0.7] {};
\node at (0,-1) [circle,draw=black,fill=black,scale=0.7] {};
\node at (1,-1) [circle,draw=black,fill=black,scale=0.7] {};
\node at (2,-1) [circle,draw=black,fill=black,scale=0.7] {};
\node at (2.6,-1) [circle,draw=black,fill=black,scale=0.7] {};
\node at (3.4,-1) [circle,draw=black,fill=black,scale=0.7] {};
\node at (4,-1) [circle,draw=black,fill=black,scale=0.7] {};
\node at (5,-1) [circle,draw=black,fill=black,scale=0.7] {};
\node at (6,-1) [circle,draw=black,fill=black,scale=0.7] {};
\node at (6.6,-1) [circle,draw=black,fill=black,scale=0.7] {};
\node at (7.4,-1) [circle,draw=black,fill=black,scale=0.7] {};
\node at (8,-1) [circle,draw=black,fill=black,scale=0.7] {};
\node at (1,-0.5) [circle,draw=black,fill=black,scale=0.7] {};
\node at (5,-0.5) [circle,draw=black,fill=black,scale=0.7] {};
\node at (-2,-0.6){$\ell'_{1}$};
\node at (8,-0.6){$\ell'_{2}$};
\node at (-1,-1.5){$\underbrace{\ \ \ .\  .\  .\ \ \ }_{k-2}$};
\node at (3,-1.5){$\underbrace{\ \ \ .\  .\  .\ \ \ }_{2(k-2)}$};
\node at (7,-1.5){$\underbrace{\ \ \ .\  .\  .\ \ \ }_{k-2}$};
\end{tikzpicture}
\end{center}
\caption{The graph $L_k$ (on the top) and the graph $J_k$ (on the bottom).}
\label{fJK}
\end{figure}
\begin{figure}[t]
\begin{center}
\begin{tikzpicture}
\draw (0,0) -- (0.5,0);
\draw (1.1,0) -- (2.9,0);
\draw (3.5,0) -- (4.8,0);
\draw[dotted] (0.5,0) -- (1.1,0);
\draw[dotted] (2.9,0) -- (3.5,0);
\draw (0,0) -- (0.65,0.65);
\draw (0.95,0.95) -- (1.85,1.85);
\draw (2.15,2.15) -- (2.4,2.4);
\draw[dotted] (0.65,0.65) -- (0.95,0.95);
\draw[dotted] (1.85,1.85) -- (2.15,2.15);
\draw (4.8,0) -- (4.55,0.25);
\draw (4.25,0.55) -- (3.35,1.45);
\draw (3.05,1.75) -- (2.4,2.4);
\draw[dotted] (4.55,0.25) -- (4.25,0.55);
\draw[dotted] (3.35,1.45) -- (3.05,1.75);
\draw (1.6,0) -- (2,-0.5);
\draw (2.4,0) -- (2,-0.5);
\draw (1.6,1.6) -- (1,1.8);
\draw (1.2,1.2) -- (1,1.8);
\draw (4,0.8) -- (4.2,1.4);
\draw (3.6,1.2) -- (4.2,1.4);
\draw (1,0.44) -- (0.4,0.4);
\draw (2,0.6) -- (1.55,0.55);
\draw[dotted] (1,0.44) -- (1.55,0.55);
\draw (3.6,0.2) -- (4,0);
\draw (2.8,0.6) -- (3.2,0.4);
\draw[dotted] (3.6,0.2) -- (3.2,0.4);
\draw (2.65,1.75) -- (2.8,2);
\draw (2.4,1.2)  -- (2.55,1.45);
\draw[dotted] (2.65,1.75)  -- (2.55,1.45);
\draw (2,0.6) -- (2.8,0.6);
\draw (2,0.6) -- (2.4,1.2);
\draw (2.8,0.6) -- (2.4,1.2);
\node at (0,0) [circle,draw=black,fill=black,scale=0.7] {};
\node at (0.5,0) [circle,draw=black,fill=black,scale=0.7] {};
\node at (1.1,0) [circle,draw=black,fill=black,scale=0.7] {};
\node at (1.6,0) [circle,draw=black,fill=black,scale=0.7] {};
\node at (2.4,0) [circle,draw=black,fill=black,scale=0.7] {};
\node at (2.9,0) [circle,draw=black,fill=black,scale=0.7] {};
\node at (3.5,0) [circle,draw=black,fill=black,scale=0.7] {};
\node at (4,0) [circle,draw=black,fill=black,scale=0.7] {};
\node at (4.8,0) [circle,draw=black,fill=black,scale=0.7] {};
\node at (2,-0.5) [circle,draw=black,fill=black,scale=0.7] {};
\node at (0.4,0.4) [circle,draw=black,fill=black,scale=0.7] {};
\node at (0.65,0.65) [circle,draw=black,fill=black,scale=0.7] {};
\node at (0.95,0.95) [circle,draw=black,fill=black,scale=0.7] {};
\node at (1.2,1.2) [circle,draw=black,fill=black,scale=0.7] {};
\node at (1.6,1.6) [circle,draw=black,fill=black,scale=0.7] {};
\node at (1.85,1.85) [circle,draw=black,fill=black,scale=0.7] {};
\node at (2.15,2.15) [circle,draw=black,fill=black,scale=0.7] {};
\node at (2.4,2.4) [circle,draw=black,fill=black,scale=0.7] {};
\node at (4.55,0.25) [circle,draw=black,fill=black,scale=0.7] {};
\node at (4.25,0.55) [circle,draw=black,fill=black,scale=0.7] {};
\node at (4,0.8) [circle,draw=black,fill=black,scale=0.7] {};
\node at (3.6,1.2) [circle,draw=black,fill=black,scale=0.7] {};
\node at (3.35,1.45) [circle,draw=black,fill=black,scale=0.7] {};
\node at (3.05,1.75) [circle,draw=black,fill=black,scale=0.7] {};
\node at (2.8,2) [circle,draw=black,fill=black,scale=0.7] {};
\node at (4.2,1.4) [circle,draw=black,fill=black,scale=0.7] {};
\node at (1,1.8) [circle,draw=black,fill=black,scale=0.7] {};
\node at (2,0.6) [circle,draw=black,fill=black,scale=0.7] {};
\node at (2.8,0.6) [circle,draw=black,fill=black,scale=0.7] {};
\node at (2.4,1.2) [circle,draw=black,fill=black,scale=0.7] {};
\node at (1,0.44) [circle,draw=black,fill=black,scale=0.7] {};
\node at (1.55,0.55) [circle,draw=black,fill=black,scale=0.7] {};
\node at (3.6,0.2) [circle,draw=black,fill=black,scale=0.7] {};
\node at (3.2,0.4) [circle,draw=black,fill=black,scale=0.7] {};

\node at (2.55,1.45) [circle,draw=black,fill=black,scale=0.7] {};
\node at (2.65,1.75) [circle,draw=black,fill=black,scale=0.7] {};

\node at (0.8,-0.55){$\underbrace{\ \ \ .\  .\  .\ \ \ }_{k-2}$};
\node at (3.2,-0.55){$\underbrace{\ \ \ .\  .\  .\ \ \ }_{k-2}$};
\node at (4.8,0.8)[rotate=-45]{$\overbrace{\ .\  .\  .\  }^{k-2}$};
\node at (3.6,2)[rotate=-45]{$\overbrace{\ .\  .\  .\  }^{k-2}$};
\node at (0.4,1.2)[rotate=45]{$\overbrace{\ .\  .\  .\  }^{k-2}$};
\node at (1.6,2.4)[rotate=45]{$\overbrace{\ .\  .\  .\  }^{k-2}$};
\node at (3.6,0.65) [rotate=-32]{$\overbrace{. . . }^{k-2}$};
\node at (1.4,0.85) [rotate=10]{$\overbrace{ .  . .  }^{k-2}$};
\node at (2.1,1.6) [rotate=65]{$\overbrace{. . .  }^{k-2}$};
\node at (-0.4,0) {$\ell_{1}$};
\node at (5.2,0) {$\ell_{2}$};
\node at (2.4,2.8) {$\ell_{3}$};

\draw (0+6.3,0) -- (4.8+6.3,0);
\draw (0+6.3,0) -- (2.4+6.3,2.4);
\draw (4.8+6.3,0) -- (2.4+6.3,2.4);
\draw (1.6+6.3,0) -- (2+6.3,-0.5);
\draw (2.4+6.3,0) -- (2+6.3,-0.5);
\draw (1.6+6.3,1.6) -- (1+6.3,1.8);
\draw (1.2+6.3,1.2) -- (1+6.3,1.8);
\draw (4+6.3,0.8) -- (4.2+6.3,1.4);
\draw (3.6+6.3,1.2) -- (4.2+6.3,1.4);
\draw (2+6.3,0.6) -- (0.4+6.3,0.4);
\draw (2.8+6.3,0.6) -- (4+6.3,0);
\draw (2.4+6.3,1.2)  -- (2.8+6.3,2);
\draw (2+6.3,0.6) -- (2.8+6.3,0.6);
\draw (2+6.3,0.6) -- (2.4+6.3,1.2);
\draw (2.8+6.3,0.6) -- (2.4+6.3,1.2);
\node at (0+6.3,0) [circle,draw=blue!50,fill=blue!20] {};
\node at (0+6.3,0) {$1_a$};
\node at (0.8+6.3,0) [circle,draw=red!50,fill=red!20] {};
\node at (0.8+6.3,0) {$1_b$};
\node at (1.6+6.3,0) [circle,draw=blue!50,fill=blue!20] {};
\node at (1.6+6.3,0) {$1_a$};
\node at (2.4+6.3,0) [circle,draw=red!50,fill=red!20] {};
\node at (2.4+6.3,0) {$1_b$};
\node at (3.2+6.3,0) [circle,draw=blue!50,fill=blue!20] {};
\node at (3.2+6.3,0) {$1_a$};
\node at (4+6.3,0) [circle,draw=red!50,fill=red!20] {};
\node at (4+6.3,0) {$1_b$};
\node at (4.8+6.3,0) [circle,draw=blue!50,fill=blue!20] {};
\node at (4.8+6.3,0) {$1_a$};
\node at (2+6.3,-0.5)  [circle,draw=green!50,fill=green!20] {};
\node at (2+6.3,-0.5) {3};
\node at (0.4+6.3,0.4) [circle,draw=red!50,fill=red!20] {};
\node at (0.8,-0.55){$\underbrace{\ \ \ .\  .\  .\ \ \ }_{k-2}$};
\node at (0.4+6.3,0.4) {$1_b$};
\node at (0.8+6.3,0.8) [circle,draw=blue!50,fill=blue!20] {};
\node at (0.8+6.3,0.8) {$1_a$};
\node at (1.2+6.3,1.2) [circle,draw=red!50,fill=red!20] {};
\node at (1.2+6.3,1.2) {$1_b$};
\node at (1.6+6.3,1.6) [circle,draw=green!50,fill=green!20] {};
\node at (1.6+6.3,1.6) {3};
\node at (2+6.3,2) [circle,draw=blue!50,fill=blue!20] {};
\node at (2+6.3,2) {$1_a$};
\node at (2.4+6.3,2.4) [circle,draw=red!50,fill=red!20] {};
\node at (2.4+6.3,2.4) {$1_b$};
\node at (4.4+6.3,0.4) [circle,draw=red!50,fill=red!20] {};
\node at (4.4+6.3,0.4) {$1_b$};
\node at (4+6.3,0.8) [circle,draw=green!50,fill=green!20] {};
\node at (4+6.3,0.8) {3};
\node at (3.6+6.3,1.2) [circle,draw=blue!50,fill=blue!20] {};
\node at (3.6+6.3,1.2) {$1_a$};
\node at (3.2+6.3,1.6) [circle,draw=red!50,fill=red!20] {};
\node at (3.2+6.3,1.6) {$1_b$};
\node at (2.8+6.3,2) [circle,draw=blue!50,fill=blue!20] {};
\node at (2.8+6.3,2) {$1_a$};
\node at (4.2+6.3,1.4) [circle,draw=red!50,fill=red!20] {};
\node at (4.2+6.3,1.4) {$1_b$};
\node at (1+6.3,1.8) [circle,draw=blue!50,fill=blue!20] {};
\node at (1+6.3,1.8){$1_a$};
\node at (2+6.3,0.6)[circle,draw=green!50,fill=green!20] {};
\node at (2+6.3,0.6){3};
\node at (2.8+6.3,0.6) [circle,draw=red!50,fill=red!20] {};
\node at (2.8+6.3,0.6) {$1_b$};
\node at (2.4+6.3,1.2)  [circle,draw=blue!50,fill=blue!20] {};
\node at (2.4+6.3,1.2) {$1_a$};
\node at (1.25+6.3,0.5) [circle,draw=blue!50,fill=blue!20] {};
\node at (1.25+6.3,0.5) {$1_a$};
\node at (3.4+6.3,0.3) [circle,draw=blue!50,fill=blue!20] {};
\node at (3.4+6.3,0.3) {$1_a$};
\node at (2.6+6.3,1.6) [circle,draw=red!50,fill=red!20] {};
\node at (2.6+6.3,1.6) {$1_b$};
\node at (-0.35+6.3,0) {$\ell_{1}$};
\node at (5.2+6.3,0) {$\ell_{2}$};
\node at (2.4+6.3,2.8) {$\ell_{3}$};
\end{tikzpicture}
\end{center}
\caption{The graph $H_k$ (on the left) and a $(1,1,3)$-coloring of $H_3$ (on the right).}
\label{fH}
\end{figure}
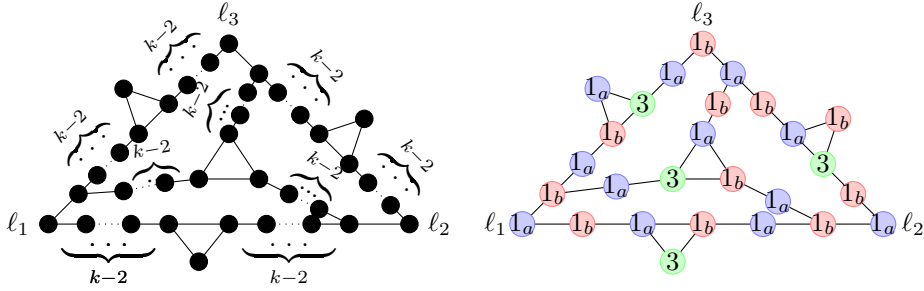
\begin{proof}
Let $L_k$ and $J_k$ be the subcubic graphs from Figure~\ref{fJK} with the corresponding connectors $\ell'_1$ and $\ell'_2$.
First, observe that there are two triangles in $L_k$.
If no common vertices of the two triangles are colored by $k$, then the coloring can not be extended to a $(1,1,k)$-coloring of $L_k$. Thus, in every $(1,1,k)$-coloring one of these two vertices is colored by $k$ and the remaining vertices could only be colored by $1_a$ or by $1_b$, alternating the colors. Therefore, since the number of the remaining vertices is odd, if one of the connector is colored by a color $1$, the other connector could only be colored by the same color $1$.
In this $(1,1,k)$-coloring, the vertices colored by $k$ are at distance $k$ from the connectors.

Observe that there are two pairs of adjacent triangles in $J$.
Thus, for the same reason than for $L$, for each pair of adjacent triangles, one of the two vertices which is common to both triangles should be colored by $k$. Therefore, since the number of remaining vertices is even, if one of the connector is colored by a color $1$, the other connector could only be colored by the other color $1$.
To conclude, $L_k$ is a 1-1-transmitter and $J_k$ is a 1-1-antitransmitter.

Let $H_k$ be the subcubic graph from the left part of Figure~\ref{fH} with the corresponding connectors $\ell_1$, $\ell_2$ and $\ell_3$. The graph $H_k$ contains four triangles and in every triangle a vertex should be colored by $k$. Thus, in every $(1,1,k)$-coloring the remaining vertices could only be colored by $1_a$ or by $1_b$ alternating the colors.
Without loss of generality, suppose $\ell_1$, $\ell_2$ and $\ell_3$ are colored by the same color $1_a$. The two neighbors of these vertices should be colored by $1_b$ and alternating the colors, the central triangle has its neighbors all colored by the same color $1_b$ or by $1_a$ (depending on $k$). Hence, the central triangle can not be colored, the color of the neighbors being not available.
The right part of Figure~\ref{fH} gives the existence of an $S$-coloring with the suited property for $k=3$ (for any other integer $k$, the coloring is similar).
To conclude, $H_k$ is a 1-1-clause simulator.
\end{proof}
\begin{cor}\label{cor1}
Let $k\ge2$ be a positive integer.
The problem $(1,1,k)$-COL is NP-complete for cubic graphs.
\end{cor}
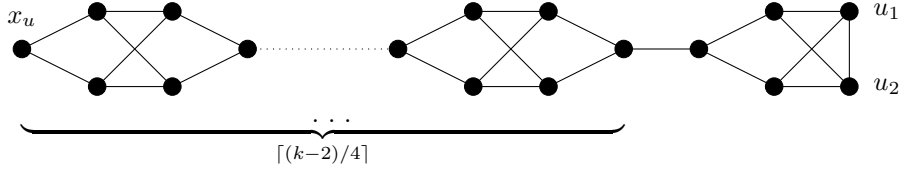
\begin{figure}[t]
\begin{center}
\begin{tikzpicture}
\draw (0,0) -- (1,0.5);
\draw (0,0) -- (1,-0.5);
\draw (1,0.5) -- (2,0.5);
\draw (1,-0.5) -- (2,-0.5);
\draw (1,0.5) -- (2,-0.5);
\draw (1,-0.5) -- (2,0.5);
\draw (2,0.5) -- (3,0);
\draw (2,-0.5) -- (3,0);
\draw[dotted] (3,0) -- (5,0);
\draw (5,0) -- (6,0.5);
\draw (5,0) -- (6,-0.5);
\draw (6,0.5) -- (7,0.5);
\draw (6,-0.5) -- (7,-0.5);
\draw (6,0.5) -- (7,-0.5);
\draw (6,-0.5) -- (7,0.5);
\draw (7,0.5) -- (8,0);
\draw (7,-0.5) -- (8,0);
\draw (8,0) -- (9,0);
\draw (9,0) -- (10,0.5);
\draw (9,0) -- (10,-0.5);
\draw (10,0.5) -- (11,0.5);
\draw (10,-0.5) -- (11,-0.5);
\draw (10,0.5) -- (11,-0.5);
\draw (10,-0.5) -- (11,0.5);
\draw (11,-0.5) -- (11,0.5);
\node at (4,-1.25){$\underbrace{ \ \ \ \ \ \ \ \ \ \ \ \ \ \ \ \ \ \ \ \ \ \ \ \ \ \ \ \ \ \ \ \ \ . \ . \ . \ \ \ \ \ \ \ \ \ \ \ \ \ \ \ \ \ \ \ \ \ \ \ \ \ \ \ \ \ \ \ }_{\lceil (k-2)/4\rceil}$};
\node at (0,0) [circle,draw=black,fill=black, scale=0.7] {};
\node at (1,0.5) [circle,draw=black,fill=black, scale=0.7] {};
\node at (1,-0.5) [circle,draw=black,fill=black, scale=0.7] {};
\node at (2,0.5) [circle,draw=black,fill=black, scale=0.7] {};
\node at (2,-0.5) [circle,draw=black,fill=black, scale=0.7] {};
\node at (3,0) [circle,draw=black,fill=black, scale=0.7] {};
\node at (5,0) [circle,draw=black,fill=black, scale=0.7] {};
\node at (6,0.5) [circle,draw=black,fill=black, scale=0.7] {};
\node at (6,-0.5) [circle,draw=black,fill=black, scale=0.7] {};
\node at (7,0.5) [circle,draw=black,fill=black, scale=0.7] {};
\node at (7,-0.5) [circle,draw=black,fill=black, scale=0.7] {};
\node at (8,0) [circle,draw=black,fill=black, scale=0.7] {};
\node at (9,0) [circle,draw=black,fill=black, scale=0.7] {};
\node at (10,0.5) [circle,draw=black,fill=black, scale=0.7] {};
\node at (10,-0.5) [circle,draw=black,fill=black, scale=0.7] {};
\node at (11,0.5) [circle,draw=black,fill=black, scale=0.7] {};
\node at (11,-0.5) [circle,draw=black,fill=black, scale=0.7] {};
\node at (0,0.4){$x_u$};
\node at (11.5,0.5){$u_1$};
\node at (11.5,-0.5){$u_2$};
\end{tikzpicture}
\end{center}
\caption{The graph $N_k$.}
\label{fN}
\end{figure}
\begin{proof}
Let $G$ be a graph and $G'$ be the graph obtained with the reduction from Lemma \ref{reffinal} using the 1-1-transmitter, 1-1-antitransmitter and 1-1-clause simulator from Proposition \ref{reffinal2}.
For each vertex $u\in V(G')$ of degree 2, we add a graph $N_k$ from Figure \ref{fN} and connect $u$ to the vertex $x_u$ of $N_k$. 
We obtain a cubic graph in the previous reduction. The graph $N_k$ is clearly $(1,1,k)$-colorable, the color $k$ is used to color one vertex among $u_1$ and $u_2$.
\end{proof}
\begin{proof}[Proof of Theorem \ref{theo22}]
Theorem \ref{theo22} is obtained from Theorem \ref{theo23}, Theorem \ref{theo24} and Corollary \ref{cor1}.
\end{proof}
\section{Dichothomy theorem for $|S|= 4$}
This section is dedicated to the proof of our main theorem, a dichotomy property for $S$-COL when $|S|= 4$.
\begin{theo}\label{theo31}
Let $S$ be a nondecreasing list of four integers. The problem $S$-COL is polynomial-time solvable if $S\le S'$, for $S'=(2,3,3,3)$, $S'=(2,2,3,4)$, $S'=(1,4,4,4)$ or $S'=(1,2,5,6)$, and NP-complete otherwise.
\end{theo}
The section is organized as follows. In the first subsection, we present some properties used to prove that $S$-COL is polynomial-time solvable. In the second subsection, we present some properties used to prove that $S$-COL is NP-complete.
In the third subsection, we prove Theorem \ref{theo31}.
\subsection{Polynomial-time solvable instances of $S$-COL}

The following theorem is a result from Bienstock \emph{et al.} \cite{BI1991}.
The definitions of pathwidth, treewidth and minors are not given here, the reader can found them in \cite{RO1984,WA1937}.
Note that for a graph $G$, the treewidth of $G$ is bounded by the pathwidth of $G$.
\begin{theo}[\cite{BI1991}]\label{treewborne}
For any forest $F$, every graph without $F$ as minor has a pathwidth bounded by $|V(F)|-2$.
\end{theo}
We will use the following theorem proven by Fiala \emph{et al.} \cite{FI2010} in order to prove that several instances of $S$-COL are polynomial-time solvable.
\begin{theo}[\cite{FI2010}]\label{Spoly}
Let $S$ be a nondecreasing list of integers. The problem $S$-COL is polynomial-time solvable for graphs of bounded treewidth.
\end{theo}
This result has been generalized to the decision problem associated with the $S$-$k$-coloring problem (with a graph and $k$ as instances), for a sequence $S$ which contains only integers smaller than a constant \cite{FI2010}.
The theorem from Fiala \emph{et al.} \cite{FI2010} uses the result of Courcelle \cite{CO1990} about properties definable in monadic second-order logic.

Let $Y_n$ denotes the graph obtained from taking three copies of $P_{n+1}$ and identifying one end.
We start with the following observation:
\begin{obs}\label{ynpascol}
Let $G$ be a graph and $n$ be a positive integer. If $G$ contains $Y_n$ as minor, then $G$ contains the graph $Y_n$ as subgraph.
\end{obs}
\begin{prop}
If there exist a positive integer $n$ such that $Y_{n}$ is not $S$-colorable, then $S$-COL is polynomial-time solvable.
\end{prop}
\begin{proof}
 We can note that, by Observation \ref{ynpascol}, every graph containing the graph $Y_{n}$ as minor is not $S$-colorable. Using Theorem \ref{treewborne}, every $S$-colorable graph has treewidth bounded by $3n-1$. Using Theorem \ref{Spoly}, $S$-COL is polynomial-time solvable.
\end{proof}

We recall the following theorem from Sharp \cite{ASH2012} about the complexity of the distance coloring in $k$ colors ($S_{d^k}$ is the list composed of $k$ integers $d$).
\begin{theo}[\cite{ASH2012}]\label{sharpp}
The problem $S_{d^k}$-COL is NP-complete for $k\ge \lceil 3d/2 \rceil$ and polynomial-time solvable in the other cases.
\end{theo}
\begin{prop}\label{dichoa}
Let $S$ be a nondecreasing list of integers and let $i>1$ be an integer. If $s_1=i$, then $S$-COL is polynomial-time solvable for $|S|<\lceil 3i/2 \rceil$.
\end{prop}
\begin{proof}
Using Theorem \ref{sharpp}, the problem $S_{d^k}$-COL is polynomial-time solvable for $k<\lceil 3d/2 \rceil$.
Every $S_{i^{|S|}}$-colorable graphs has bounded treewidth \cite{ASH2012} and every $S$-colorable graph is also $S_{i^{|S|}}$-colorable. Hence, $S$-colorable graphs have bounded treewidth. Using Theorem \ref{Spoly}, $S$-COL is polynomial-time solvable.
\end{proof}
\subsection{NP-complete instances of $S$-COL for $|S|\ge4$}

The following proposition is a generalization of the Theorem from Sharp \cite{ASH2012}.
\begin{theo}
The problem $S$-COL is NP-complete for every list $S$ such that $s_1=s_{i+1}=i$, with $i>1$ and $|S|\ge \lceil 3i/2 \rceil$.
\label{reducii}
\end{theo}
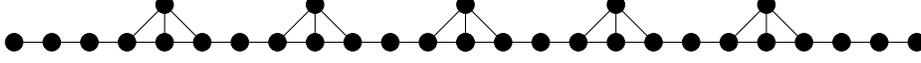
\begin{figure}[t]
\begin{center}
\begin{tikzpicture}
\draw (-12*0.5,0) -- (12*0.5,0);
\draw (0,0) -- (0,0.5);
\draw (-0.5,0) -- (0,0.5);
\draw (0.5,0) -- (0,0.5);
\draw (2,0) -- (2,0.5);
\draw (1.5,0) -- (2,0.5);
\draw (2.5,0) -- (2,0.5);
\draw (4,0) -- (4,0.5);
\draw (3.5,0) -- (4,0.5);
\draw (4.5,0) -- (4,0.5);
\draw (-2,0) -- (-2,0.5);
\draw (-1.5,0) -- (-2,0.5);
\draw (-2.5,0) -- (-2,0.5);
\draw (-4,0) -- (-4,0.5);
\draw (-3.5,0) -- (-4,0.5);
\draw (-4.5,0) -- (-4,0.5);
\foreach \x in {0,...,12}
{
\node at (\x * 0.5 ,0) [circle,draw=black,fill=black,scale=0.7] {};
\node at (-\x * 0.5 ,0) [circle,draw=black,fill=black,scale=0.7] {};
}
\node at (0,0.5) [circle,draw=black,fill=black,scale=0.7] {};
\node at (2,0.5) [circle,draw=black,fill=black,scale=0.7] {};
\node at (4,0.5) [circle,draw=black,fill=black,scale=0.7] {};
\node at (-2,0.5) [circle,draw=black,fill=black,scale=0.7] {};
\node at (-4,0.5) [circle,draw=black,fill=black,scale=0.7] {};
\end{tikzpicture}
\end{center}
\caption{The graph $L_{2,(2,2,2,3)}$.}
\label{gadgetgrossereduc}
\end{figure}
\begin{figure}[t]
\begin{center}
\begin{tikzpicture}
\draw (0,0) -- (0,0.4);
\draw[dotted] (0,0.4) -- (0,0.8);
\draw (0,0.8) -- (0,1.2);
\draw (0,0) -- (-0.83,-0.83);
\draw[dotted] (-0.83,-0.83) -- (-1.16,-1.16);
\draw (-1.16,-1.16) -- (-1.5,-1.5);
\draw (0,0) -- (0.83,-0.83);
\draw[dotted] (0.83,-0.83) -- (1.16,-1.16);
\draw (1.16,-1.16) -- (1.5,-1.5);
\draw (0.5,-0.5) -- (-0.5,-0.5);
\draw[dotted] (0,0) -- (-0.33,-0.83);
\draw[dotted] (0,0) -- (0,-0.83);
\draw[dotted] (0,0) -- (0.33,-0.83);
\draw[dotted] (0.5,-0.5) -- (-0.33,-0.83);
\draw[dotted] (0.5,-0.5) -- (0,-0.83);
\draw[dotted] (0.5,-0.5) -- (0.33,-0.83);
\draw[dotted] (-0.5,-0.5) -- (-0.33,-0.83);
\draw[dotted] (-0.5,-0.5) -- (0,-0.83);
\draw[dotted] (-0.5,-0.5) -- (0.33,-0.83);
\node at (0,0) [circle,draw=black,fill=black, scale=0.5] {};
\node at (0,0.4) [circle,draw=black,fill=black, scale=0.5] {};
\node at (0,0.8) [circle,draw=black,fill=black, scale=0.5] {};
\node at (0,1.2) [circle,draw=black,fill=black, scale=0.5] {};
\node at (-0.5,-0.5) [circle,draw=black,fill=black, scale=0.5] {};
\node at (-0.83,-0.83) [circle,draw=black,fill=black, scale=0.5] {};
\node at (-1.16,-1.16) [circle,draw=black,fill=black, scale=0.5] {};
\node at (-1.5,-1.5) [circle,draw=black,fill=black, scale=0.5] {};
\node at (0.5,-0.5) [circle,draw=black,fill=black, scale=0.5] {};
\node at (0.83,-0.83) [circle,draw=black,fill=black, scale=0.5] {};
\node at (1.16,-1.16) [circle,draw=black,fill=black, scale=0.5] {};
\node at (1.5,-1.5) [circle,draw=black,fill=black, scale=0.5] {};
\node at (-0.33,-0.83) [circle,draw=black,fill=black, scale=0.5] {};
\node at (0,-0.83) [circle,draw=black,fill=black, scale=0.5] {};
\node at (0.33,-0.83) [circle,draw=black,fill=black, scale=0.5] {};
\node at (0,-1.2){$\underbrace{... }$};
\node at (0,-1.5){$|S|-\lceil 3i/2 \rceil$};
\node at (-0.4,0.6)[rotate=90]{$\overbrace{... }$};
\node at (-1.35,-0.75)[rotate=45]{$\overbrace{... }$};
\node at (1.35,-0.75)[rotate=-45]{$\overbrace{... }$};
\node at (-1.2,0.6){$i/2-1$};
\node at (-2.1,-0.6){$i/2-1$};
\node at (2.1,-0.6){$i/2-1$};

\draw (0+6,0) -- (0+6,0.4);
\draw[dotted] (0+6,0.4) -- (0+6,0.8);
\draw (0+6,0.8) -- (0+6,1.2);
\draw (0+6,0) -- (-0.33+6,-0.33);
\draw[dotted] (-0.66+6,-0.66) -- (-0.33+6,-0.33);
\draw (-0.66+6,-0.66) -- (-1+6,-1);
\draw (0+6,0) -- (0.33+6,-0.33);
\draw[dotted] (0.66+6,-0.66) -- (0.33+6,-0.33);
\draw (0.66+6,-0.66) -- (1+6,-1);
\draw[dotted] (0+6,0) -- (-0.33+6,-0.83);
\draw[dotted] (0+6,0) -- (0+6,-0.83);
\draw[dotted] (0+6,0) -- (0.33+6,-0.83);
\node at (-0.5,1.2){$c_1$};
\node at (-2,-1.5){$c_2$};
\node at (2,-1.5){$c_3$};
\node at (0+6,0) [circle,draw=black,fill=black, scale=0.5] {};
\node at (0+6,0.4) [circle,draw=black,fill=black, scale=0.5] {};
\node at (0+6,0.8) [circle,draw=black,fill=black, scale=0.5] {};
\node at (0+6,1.2) [circle,draw=black,fill=black, scale=0.5] {};
\node at (-1+6,-1) [circle,draw=black,fill=black, scale=0.5] {};
\node at (-0.66+6,-0.66) [circle,draw=black,fill=black, scale=0.5] {};
\node at (-0.33+6,-0.33) [circle,draw=black,fill=black, scale=0.5] {};
\node at (0.66+6,-0.66) [circle,draw=black,fill=black, scale=0.5] {};
\node at (0.33+6,-0.33) [circle,draw=black,fill=black, scale=0.5] {};
\node at (1+6,-1) [circle,draw=black,fill=black, scale=0.5] {};
\node at (-0.33+6,-0.83) [circle,draw=black,fill=black, scale=0.5] {};
\node at (0+6,-0.83) [circle,draw=black,fill=black, scale=0.5] {};
\node at (0.33+6,-0.83) [circle,draw=black,fill=black, scale=0.5] {};
\node at (0+6,-1.2){$\underbrace{... }$};
\node at (0+6,-1.5){$|S|-\lceil 3i/2 \rceil$};
\node at (-0.5+6,0.6)[rotate=90]{$\overbrace{... }$};
\node at (-0.95+6,-0.25)[rotate=45]{$\overbrace{... }$};
\node at (0.95+6,-0.25)[rotate=-45]{$\overbrace{... }$};
\node at (-1.2+6,0.6){$\lfloor i/2 \rfloor$};
\node at (-1.7+6,-0.1){$\lfloor i/2 \rfloor$};
\node at (1.7+6,-0.1){$\lfloor i/2 \rfloor$};
\node at (-0.5+6,1.2){$c_1$};
\node at (-1.5+6,-1){$c_2$};
\node at (1.5+6,-1){$c_3$};
\end{tikzpicture}
\end{center}
\caption{The graphs $R_{i,S}$, for an even $i$ (on the left) and for an odd $i$ (on the right).}
\label{Ri}
\end{figure}
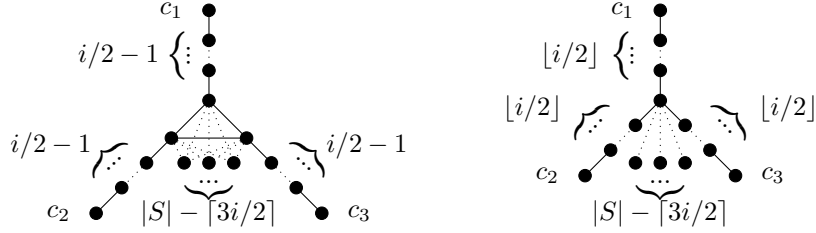
\begin{proof}
Let $N_{i}(S)$ be the number of occurrences of an integer $i$ in the list $S$ and let $n_{i}^{-}(S)=N_i(S)-i-1$.
The proof uses a reduction from $N_i(S)$-COL (as $N_i(S)\ge 3$, $N_i(S)$-COL is NP-complete).
\begin{center}
\parbox{10cm}{
\setlength{\parskip}{.05cm}
\textbf{$k$-COL}

\textbf{Input}: A graph $G$ and an integer $k$. 

\textbf{Question}: Is the graph $G$ $k$-colorable?}
\end{center}

Let $G$ be a graph and $A(S)$ be the following set of integers $\{j>0 |\ |N_j(S)|\ge 1 \}$.
We set $m=\prod_{ j\in A(S) } (j+1)$.
We construct a graph $L_{i,S}$ composed of vertices $u_{j,l}$, for integers $j$ and $l$, with $-m \le j \le m$.
Let $j$, $j'$, $l$ and $l'$ be four integers. Two vertices $u_{j,l}$ and $u_{j',l'}$ in $V(L_{i,S})$ are adjacent if and only if $|j-j'|\le 1$ (except if $u_{j,l}=u_{j',l'}$).
The rank of a vertex is $\text{rank}(u_{j,l})=j$ and corresponds to its position in $L_{i,S}$.
Let $j$ be an integer and let $\text{rank}_{max} (j)=\max\limits \{l\ |\ u_{j,l}\in V(L_{i,S})\}$.
Let $k$ be an integer, a $k$-insertion at rank $j$ consists in adding vertices $u_{j,\text{rank}_{max} (i)+1}, \ldots , u_{j,\text{rank}_{max} (i)+k}$ in the graph $L_{i,S}$.

We construct the graph $L_{i,S}$ as follows:
the graph $L_{i,S}$ always contains vertices $u_{-m,0},\ldots ,u_{m,0}$.
Depending on the list $S$, we have to do the following insertions: 
\begin{enumerate}
\item We do an $n_{i}^{-}(S)$-insertion at rank $j$, for each integer $j$ such that $j\equiv \lfloor (i+1)/ 2 \rfloor\pmod{i+1}$ and $-m \le j\le m$.
\item For every integer $k>i$, we do an $N_{k}(S)$-insertion at rank $j$, for each integer $j$ such that $j\equiv 0\pmod{k+1}$ and $-m < j< m$.
\end{enumerate}
The last operation consists in adding edges in the graph in order that the adjacency property (two vertices of rank $j$ and $j'$ are adjacent if and only if $|j-j'|\le 1$) is satisfied.
We define \emph{connectors} as specified vertices of degree 1.
The connectors of $L_{i,S}$ are the vertices $u_{-m,0}$ and $u_{m,0}$. The graph $L_{2,(2,2,2,3)}$ is illustrated in Figure \ref{gadgetgrossereduc}.

Let $R_{i,S}$ be the graph from Figure \ref{Ri}. Note that depending on the parity of $i$, the graph $R_{i,S}$ is different. The connectors of $R_{i,S}$ are the vertices $c_1$, $c_2$ and $c_3$ from Figure \ref{Ri}.
Connecting a graph $G$ to another graph $G'$ is an operation that consists in identifying a connector of degree 1 of $G$ with a connector of degree 1 of $G'$.
In order to construct the graph $J_{i,S}$, we begin to construct a graph from two copies of $L_{i,S}$ and a copy of $R_{i,S}$ where we connect $R_{i,S}$ to the two copies of $L_{i,S}$. Suppose that $u$ and $v$ are the two connectors of degree 1 of the copies of $L_{i,S}$ in the constructed graph. We finish by deleting one vertex among $u$ and $v$. If $i=2$, the following property should be ensured by the copy of $L_{i,S}$: $\text{rank}_{max}(\text{rank}(w))=1$, for $w$ the vertex adjacent to the removed vertex. The constructed graph is the graph $J_{i,S}$. The connectors of this graph are the vertices $w$ and the remaining vertex among $u$ and $v$.

Let $d$ be a positive integer.
The graph $G_{d}$ is the graph obtained from $R_{i,S}$ and $L_{i,S}$, as follows: we begin with taking $d$ copies $L_0,\ldots, L_{d-1}$ of $L_{i,S}$ and $d$ copies $R_{0},\ldots, R_{d-1}$ of $R_{i,S}$. For every integer $k$, with $0\le k \le d-1$, we connect $R_k$ to $L_k$ and $L_{k+1\pmod{d}}$.
The obtained graph is the graph $G_d$ and have $d$ connectors, these connectors are the remaining connectors of degree 1 of each copy of $R_{i,S}$.

Finally, we construct the graph $G'$ from the original graph $G$, as follows: for every vertex $v\in V(G)$, we associate a graph $G_{v}$ which is a copy of $G_{d(v)}$.
We call original vertices, the connectors of a graph $G_v$, for a vertex $v$ (even if these vertices do not have degree 1 in the constructed graph).
For every edge $uv\in E(G)$, we connect $G_v$ and $G_u$ by a graph $J_{i,S}$.

In order to prove that the following construction is a reduction from $N_i(S)$-COL to $S$-COL, we prove the following facts:

\textbf{Fact 1}: The graph $R_{i,S}$ is $S$-colorable and in every $S$-coloring of $R_{i,S}$, the connectors of $R_{i,S}$ have the same color $i$. 

\textit{Proof}: The diameter of $R_{i,S}$ is $i+1$ and there are only three vertices which are at pairwise distance at least $i+1$, these vertices are the three connectors of $R_{i,S}$.
Since $|V(R_{i,S})|=|S|+2$, the connectors must have the same color. Moreover, this color must be a color $i$.
There exists an $S$-coloring of $R_{i,S}$ consisting in giving the same color $i$ to the connectors and giving a different color to each remaining vertex.

\textbf{Fact 2}: For every $S$-coloring of an induced $L_{i,S}$ in $G'$ and for every color $s_k\neq i$, there is a vertex of rank $j$ colored by $s_k$, for every integer $j$ such that $j\equiv 0\pmod{s_k+1}$ and $-m < j< m$.

\textit{Proof}: There are $\sum_{j=i+1}^{s_{|S|}} N_j(S) +1=|S|-n^{-}_{i}(S)$ vertices of rank $0$. Suppose that at least two vertices of rank $0$ are colored by some colors $i$.
We have to use each color $i$ two times in order to color the vertices of rank $\ell$, for $-(i+1) < \ell \le (i+1)$, since there are $|S|+N_i (S)$ vertices of such rank.
It is not possible to use each color $i$ two times if two vertices of rank $0$ are colored by some color $i$.
Thus, there are vertices of rank $0$ of each color $s_k\neq i$.

By induction on $j$, for $-m \le j\le -(i+1) $ and $(i+1) < j \le m $, we can easily prove that for every color $s_k\neq i$, there is a vertex of rank $j$ colored by $s_k$, for every integer $j$ such that $j\equiv 0\pmod{s_k+1}$.

\textbf{Fact 3}: In every $S$-coloring of $L_{i,S}$ and $J_{i,S}$, the connectors have the same color $i$ in $L_{i,S}$ and two different colors $i$ in $J_{i,S}$.

\textit{Proof}: Using Fact 2, note that in every $S$-coloring of $L_{i,S}$ and $J_{i,S}$ there is a vertex of rank $j$ colored by $s_k$, $s_k\neq i$, for every integer $j$ such that $j\equiv 0\pmod{s_k+1}$ and $-m < j< m$. Consequently, there is a vertex colored by $s_k$ at distance at most $s_k$ from the connectors of $L_{i,S}$ and $J_{i,S}$.
Note also that $-m$ and $m$ have both $0$ as remainder modulo $i+1$ in $L_{i,S}$.
There is only one vertex colored by $i$ of rank $j$, for each $j$ such that $j\equiv 0 \pmod{m}$, this color $i$ is the same for each such vertex.
Hence, the two connectors of $L_{i,S}$ have the same color $i$.
The argument is similar for $J_{i,S}$ (we have to use Fact $1$ also).

\textbf{Fact 4}: There exists an $S$-coloring of $L_{i,S}$ and $J_{i,S}$.

\textit{Proof}: First, we color the vertices $u_{0,k}$, for $-m \le k\le m$, alternating $i+1$ colors $i$.
Second, for every color $s_k\neq i$, we color the remaining vertices of rank $j$ with color $s_k$, for every integer $j$ such that $j\equiv 0\pmod{s_k+1}$ and $-m < j< m$. Finally, we color the remaining vertices of rank $j$ with $n^{-}_{i}$ remaining colors $i$, for $j\equiv \lceil i/2 \rceil\pmod{i+1}$ and $-m \le j\le m$. The coloring is similar for $J_{i,S}$.

\textbf{Fact 5}: Let $d$ be a positive integer. In every $S$-coloring of $G_d$, the original vertices have the same color $i$.

\textit{Proof}: This fact is a consequence of Facts $1$ and $3$.

\textbf{Fact 6}: For every pair of colors $(i_a,i_b)$ in $N_i(S)$, with $i_a\neq i_b$, there exists an $S$-coloring of an induced $J_{i,S}$ in $G'$, with connectors of color $i_a$ and $i_b$.

\textit{Proof}: Let $L^1$ be the copy of $L_{i,S}$ in which no vertex has been removed and let $L^2$ be the second copy of $L_{i,S}$.
We begin by using a list of $i+1$ different colors $i$ to color the vertices $u_{0,k}$, for $-m \le k\le m$, in $L^1$.
We suppose that the color $i_a$ is used to color the initial connectors of the copy of $R_{i,S}$.
We give the color $i_b$ to the neighbor of the connector of the copy of $R_{i,S}$ which is a connector of $L^2$.
We claim that we can extend the coloring to the vertices of the copy of $R_{i,S}$.
Using the colorings from Fact 4, we can also extend the coloring to the whole graph (with colors $i_a$ and $i_b$ for the connectors).

Suppose that the graph $G$ is $N_i(S)$-colorable, with colors $1,\ldots, N_i(S)$. For every vertex $v$ of color $j$, for $1\le j\le N_i(S)$, we use a color $s_j$ in order to color the original vertices of $G_v$ in the graph $G'$. Using Facts $1$, $4$ et $6$, we can extend the coloring to an $S$-coloring of the whole graph. 

Suppose that $G'$ is $S$-colorable. Using Fact $5$, for each vertex $v\in V(G)$, the original vertices of $G_v$ have the same color $i$. We conclude that we can construct an $N_i(S)$-coloring of $G$ using the color $j$ to color vertices with color $s_j$ in $G'$, for $1\le j \le N_i(S)$.
\end{proof}
\begin{cor}\label{dicho0}
Let $k$ be a positive integer.
The problem $(2,2,2,k)$-COL is NP-complete.
\end{cor}
\begin{prop}
Let $S$ be a nondecreasing list of integers with $s_2=1$ and $|S|\ge3$. If $S$-COL is NP-complete, then $S'$-COL is NP-complete for $|S'|=|S|+1$, $s'_{|S|+1}\ge s_{|S|}$ and $s'_i=s_i$, for each integer $i$, with $1\le i\le|S|$.
\label{reduc11}
\end{prop}
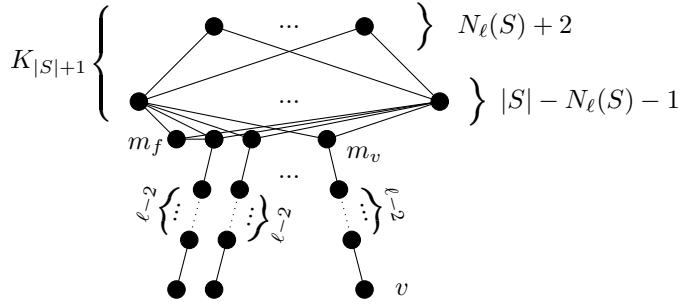
\begin{figure}[t]
\begin{center}
\begin{tikzpicture}
\draw (-1.5,0) -- (-1.33,0.66);
\draw[dotted] (-1.16,1.33) -- (-1.33,0.66);
\draw (-1.16,1.33) -- (-1,2);
\draw (-1.5,2) -- (-1,2);
\draw  (0.66,1.33)-- (0.5,2);
\draw (1,0) -- (0.83,0.66);
\draw[dotted] (0.66,1.33) -- (0.83,0.66);
\draw (-1.5+0.5,0) -- (-1.33+0.5,0.66);
\draw[dotted] (-1.16+0.5,1.33) -- (-1.33+0.5,0.66);
\draw (-1.16+0.5,1.33) -- (-1+0.5,2);
\draw (-2,2.5) -- (-1,2) ;
\draw (-2,2.5) -- (-1.5,2) ;
\draw (-2,2.5) -- (0.5,2) ;
\draw (-2,2.5) -- (-0.5,2) ;
\draw (2,2.5) -- (-1,2) ;
\draw (2,2.5) -- (-1.5,2) ;
\draw (2,2.5) -- (0.5,2) ;
\draw (2,2.5) -- (-0.5,2) ;
\draw (2,2.5) -- (1,3.5) ;
\draw (-2,2.5) -- (-1,3.5) ;
\draw (-2,2.5) -- (1,3.5) ;
\draw (2,2.5) -- (-1,3.5) ;
\node at (-1.5,0) [circle,draw=black,fill=black, scale=0.7] {};
\node at (-1,0) [circle,draw=black,fill=black, scale=0.7] {};
\node at (1,0) [circle,draw=black,fill=black, scale=0.7] {};
\node at  (-1.16,1.33) [circle,draw=black,fill=black, scale=0.7] {};
\node at (0.66,1.33) [circle,draw=black,fill=black, scale=0.7] {};
\node at (-1.16+0.5,1.33) [circle,draw=black,fill=black, scale=0.7] {};
\node at (-1.33,0.66) [circle,draw=black,fill=black, scale=0.7] {};
\node at (0.83,0.66) [circle,draw=black,fill=black, scale=0.7] {};
\node at (-1.33+0.5,0.66) [circle,draw=black,fill=black, scale=0.7] {};
\node at (-1.5,2) [circle,draw=black,fill=black, scale=0.7] {};
\node at (-1,2) [circle,draw=black,fill=black, scale=0.7] {};
\node at (0.5,2) [circle,draw=black,fill=black, scale=0.7] {};
\node at (-0.5,2) [circle,draw=black,fill=black, scale=0.7] {};
\node at (-1,2) [circle,draw=black,fill=black, scale=0.7] {};
\node at (-2,2.5) [circle,draw=black,fill=black, scale=0.7] {};
\node at (2,2.5) [circle,draw=black,fill=black, scale=0.7] {};
\node at (1,3.5) [circle,draw=black,fill=black, scale=0.7] {};
\node at (-1,3.5) [circle,draw=black,fill=black, scale=0.7] {};
\node at (0,1.5) {...};
\node at (0,2.5) {...};
\node at (0,3.5) {...};
\node at (1.25,1.1)[rotate=-75]{$\overbrace{... }^{\ell-2}$};
\node at (-1.75,1.1)[rotate=75]{$\overbrace{... }^{\ell-2}$};
\node at (-0.25,0.9)[rotate=75]{$\underbrace{... }_{\ell-2}$};
\node at (-1.9,1.9) {$m_f$};
\node at (2.5,2.5)[rotate=-90]{$\overbrace{}$};
\node at (1.8,3.5)[rotate=-90]{$\overbrace{}$};
\node at (4,2.5){$|S|-N_\ell (S)-1$};
\node at (3,3.5){$N_\ell (S)+2$};
\node at (-2.5,3)[rotate=90]{$\overbrace{\ \ \ \ \ \ \ \ \ \ \ \ \ }$};
\node at (-3.2,3){$K_{|S|+1}$};
\node at (1.5,0) {$v$};
\node at (1,1.8) {$m_v$};
\end{tikzpicture}
\end{center}
\caption{Illustration of the reduction of Proposition \ref{reduc11}.}
\label{gadg}
\end{figure}
\begin{proof}
Suppose that $S$-COL is NP-complete and that $s'_{|S|+1}=\ell$. From a graph $G$, we construct a new graph $G'$ as follows: for every vertex $v\in V(G)$, we begin with connecting $v$ to the extremity of a path of $\ell-1$ vertices, with $m_v$ the other extremity of the path (in the case $\ell=2$, the path has only one vertex and the extremities denote the same vertex). Thus, we add a new vertex $m_f$ that we connect to a vertex $m_v$, for some $v$. Afterwards, we add a copy of $K_{|S|+1}$ and we do the join between the set $\{m_v|v\in V(G)\}\cup \{ m_f\}$ and a set of $|S|-N_\ell (S)-1$ vertices in $K_{|S|+1}$. The obtained graph is the graph $G'$. The original vertices of $G'$ are the vertices that come from the graph $G$.

Suppose that $G$ is $S$-colorable. We use an $S$-coloring of $G$ in order to color the original vertices of $G'$. Thus, for every vertex $v$, we alternate the colors $s_1=1_a$ and $s_2=1_b$ in order to color every path until $m_v$. Afterwards, we give the color $1_a$ or $1_b$ to $m_f$ (depending on the color of the vertex adjacent to $m_f$). We have $N_\ell (S)+2$ vertices in $K_{|S|+1}$ non adjacent to $m_f$. The colors $1_a$, $1_b$ and the maximum number of color $\ell$ in $S$ are given to these vertices. The remaining colors are given to the non colored vertices in $K_{|S|+1}$.

Suppose that $G'$ is $S'$-colorable. Every vertex of $K_{|S|+1}$ should be colored differently. The two colors given to $m_v$ and $m_f$ should be used in order to color the $N_\ell (S)+2$ vertices in $K_{|S|+1}$ non adjacent to $m_f$. Thus, $m_v$ and $m_f$ have some colors $1$. At most $N_\ell (S)$ vertices can be colored by $\ell$ among these $N_\ell (S)+2$ vertices in $K_{|S|+1}$ non adjacent to $m_f$. Consequently, a vertex should be colored by a color $\ell$ among the remaining vertices of $K_{|S|+1}$. This color $\ell$ can not be used in order to color the other vertices of $G'$, the other vertices being at distance less than $\ell$ of a vertex colored by this color $\ell$. The coloring of the original vertices of $G'$ is used in order to create an $S$-coloring of $G$.
\end{proof}
\begin{cor}\label{dicho1}
The problem $S$-COL is NP-complete for every list $S$ such that $s_1=s_2=1$ and $|S|\ge 3$.
\end{cor}
By \emph{subdividing} an edge $uv$, we mean removing the edge $uv$, adding a new vertex and connect it to $u$ and $v$.
\begin{prop}
Let $S$ be a nondecreasing list of integers. If $S$-COL is NP-complete, then $S'$-COL is NP-complete, for every list $S'$ such that $|S'|=|S|+1$, $s'_1=1$ and $s'_{i+1}=2s_i$ or $s'_{i+1}=2s_i+1$, for every integer $i$, with $1\le i \le |S|$.
\label{reduc2f}
\end{prop}
\begin{proof}
Suppose that $S$-COL is NP-complete. There exists a reduction from $S$-COL to $S'$-COL similar with the reduction introduced by Goddard \emph{et al.} \cite{GO2008} from $(1,1,2)$-COL to $(1,2,3,4)$-COL. 

Let $G$ be a connected graph. We construct a new graph $G'$ from $G$ by replacing each edge $uv$ by $|S|+1$ parallel edges that we subdivide. The vertices $u$ and $v$ are called original vertices.

Suppose that $G$ is $S$-colorable. Let $X_{i}\subseteq V(G)$ be an $i$-packing. Every pair of vertices $(u,v)$ in $X_{i}$ is such that $d(u,v)>i$. Let $X'_{i}\subseteq V(G')$ be the set of original vertices in $X_{i}$. Observe that every pair of vertices $(u,v)$ in $X'_{i}$ is such that $d_{G'}(u,v)>2i+1$. Thus, $X'_{i}$ is a $(2i+1)$-packing (and consequently a $2i$-packing). Let $X'_{1}$ be the set of non original vertices (created by subdivision). This set is independent and we color the vertices in this set by color $1$. We give the color $s'_{i+1}$ to the original vertices of $G$ colored by $s_i$, for $1\le i\le |S|$, we obtain an $S'$-coloring of $G'$.

Suppose that $G'$ is $S'$-colorable and suppose that there exists an original vertex $u$ colored by $1$. Since $u$ has $|S|+1$ neighbors, two of its neighbors have the same color $s'_i$. Thus, no original vertex is colored by $1$. Using the $S'$-coloring of the original vertices of $G'$, we can create an $S$-coloring of $G$ by giving the color $s_{i-1}$ to the vertices colored by $s'_i$ in $G'$, for $2\le i\le |S|+1$.
\end{proof}
\begin{cor}\label{dicho2}
The problem $S$-COL is NP-complete for every list such that $|S|=4$ and $S\ge S'$, with $S'=(1,3,5,5)$.
\end{cor}
\subsection{Instances of $S$-COL for $|S|=4$}
In this section, we prove a dichotomy between NP-complete problems and polynomial-time solvable problems, using the results from the previous section. Let $a_0$ be the vertex of degree 3 in $Y_n$ and let $a_i$, $b_i$ and $c_i$ be the eventual vertices at distance $i$ from $a_0$ in $Y_n$.
\begin{prop}\label{dicho3}
Let $S$ be a nondecreasing list of integers and let $n$ be an integer. The graph $Y_{n}$ is not $S$-colorable for:
\begin{enumerate}
\item $S=(2,3,3,3)$ and $n=2$;
\item $S=(2,2,3,4)$ and $n=4$;
\item $S=(1,4,4,4)$ and $n=3$.
\end{enumerate}
\end{prop}
\begin{proof}
1. The vertices $a_0$, $a_1$, $b_1$ and $c_1$ must have different colors. Thus, one vertex among $a_2$, $b_2$ and $c_2$ can not be colored.\newline
2. The vertices $a_0$, $a_1$, $b_1$ and $c_1$ must have different colors. Thus, one vertex among $a_4$, $b_4$ and $c_4$ can not be colored.\newline
3. Suppose $a_0$ has color $1$. The three vertices $a_1$, $b_1$ and $c_1$ must have different colors. Thus, one vertex among $a_2$ and $a_3$ can not be colored. Suppose $a_0$ has a color $4$. Thus, three vertices among $a_1$, $b_1$, $c_1$, $a_2$, $b_2$ and $c_2$ must have a color $4$ and one vertex among them can not be colored anymore.
\end{proof}
\begin{prop}[\cite{GOH2012}]\label{dicho4}
The graph $P_{16}$ is not $(1,3,5,8)$-colorable.
\end{prop}
By adding a leaf on a vertex $v$, we mean adding a new vertex and connecting it to $v$.
By $\beta(G)$, we denote the order of a \emph{minimum vertex cover} of $G$.

In the following propositions, we subdivide three times an edge $uv$. By subdivide three times, we mean replacing an edge by a path of length four with three new vertices.
For an edge subdivided three times, the central subdivided vertex is the vertex which is not a neighbor of $u$ or $v$.
We introduce the following operations on a vertex of a graph, the name of these operations come from the color required to have an $S$-coloring for the introduced vertices:
\begin{enumerate}
\item A \emph{$2$-add} on a vertex $u$ is an operation which consists in adding a vertex $w_u$ and an independent set of vertices $I$, adding the edge $u w_u$ and adding some edges between $\{u, w_u\}$ and the vertices of $I$.

\item A \emph{$3$-add} on a vertex $u$ is an operation which consists in adding a vertex $w_u$ and an independent set of vertices $I$ and adding some edges between $\{u, w_u\}$ and the vertices of $I$.

\item A \emph{$3$-$5$-add} on a vertex $u$ is an operation which consists in doing a $3$-add on $u$ and doing several $3$-adds or a $2$-add on the added vertex $w_u$.

\item A \emph{$3$-$6$-add} on a vertex $u$ is an operation which consists in doing a $3$-add on $u$ and doing a $2$-add on the added vertex $w_u$.

\item A \emph{$2$-$5$-add} on a vertex $u$ is an operation which consists in doing a $2$-add on $u$ with added vertex $w_u$, doing a $2$-add on $w_u$ with added vertex $w'_u$ and doing several $3$-adds on the added vertex $w'_u$.

\end{enumerate}
These operations are illustrated in Figure \ref{figajout}.
\begin{figure}
\begin{center}
\begin{tikzpicture}
\draw (0,0) -- (0,0.7);
\draw (0,0) -- (0.5,1.5);
\draw (0,0.7) -- (0.5,1.5);
\draw (0,0) -- (1,1.5);
\draw (0,0.7) -- (-0.8,1.5);
\node at (0,0) [circle,draw=black,fill=black, scale=0.7] {};
\node at (-0.4,0.7) {$w_u$};
\node at (0,0.7) [circle,draw=red!50,fill=red!20] {};
\node at (0,0.7) {$2$};
\node at (-0.3,0) {$u$};
\node at (-0.8,1.5) [circle,draw=blue!50,fill=blue!20] {};
\node at (-0.8,1.5) {$1$};
\node at (1,1.5) [circle,draw=blue!50,fill=blue!20] {};
\node at (1,1.5) {$1$};
\node at (0.5,1.5) [circle,draw=blue!50,fill=blue!20] {};
\node at (0.5,1.5) {$1$};
\draw[dashed] (0,1.5) ellipse (1.3cm and 0.5cm);
\node at (0,1.5) {$I$};
\node at (0,-0.4) {$2$-add};
\draw (3,2) -- (2.2,1);
\draw (3,0) -- (3.5,1);
\draw (3,2) -- (3.5,1);
\draw (3,0) -- (4,1);
\node at (3,0) [circle,draw=black,fill=black, scale=0.7] {};
\node at (2.7,0) {$u$};
\node at (3,2) [circle,draw=red!50,fill=red!20] {};
\node at (3,2) {$3$};
\node at (2.6,2) {$w_u$};
\node at (2.2,1) [circle,draw=blue!50,fill=blue!20] {};
\node at (2.2,1) {$1$};
\node at (4,1) [circle,draw=blue!50,fill=blue!20] {};
\node at (4,1) {$1$};
\node at (3.5,1) [circle,draw=blue!50,fill=blue!20] {};
\node at (3.5,1) {$1$};
\node at (3,1) {$I$};
\draw[dashed] (3,1) ellipse (1.3cm and 0.5cm);
\node at (3,-0.4) {$3$-add};
\draw (6,1.75) -- (5.2,1);
\draw (6,0) -- (6.5,1);
\draw (6,1.75) -- (6.5,1);
\draw (6,0) -- (7,1);
\draw (6,1.75) -- (5.2,2.5);
\draw (6,1.75) -- (6.5,2.5);
\draw (6,3.25) -- (6.5,2.5);
\draw (6,3.25) -- (7,2.5);
\node at (6,0) [circle,draw=black,fill=black, scale=0.7] {};
\node at (5.7,0) {$u$};
\node at (6,3.25) [circle,draw=red!50,fill=red!20] {};
\node at (6,3.25) {$3$};
\node at (5.6,3.25) {$w'_u$};
\node at (6,1.75) [circle,draw=green!50,fill=green!20] {};
\node at (6,1.75) {$5$};
\node at (5.6,1.75) {$w_u$};
\node at (5.2,2.5) [circle,draw=blue!50,fill=blue!20] {};
\node at (5.2,2.5) {$1$};
\node at (7,2.5) [circle,draw=blue!50,fill=blue!20] {};
\node at (7,2.5) {$1$};
\node at (6.5,2.5) [circle,draw=blue!50,fill=blue!20] {};
\node at (6.5,2.5) {$1$};
\node at (5.2,1) [circle,draw=blue!50,fill=blue!20] {};
\node at (5.2,1) {$1$};
\node at (7,1) [circle,draw=blue!50,fill=blue!20] {};
\node at (7,1) {$1$};
\node at (6.5,1) [circle,draw=blue!50,fill=blue!20] {};
\node at (6.5,1) {$1$};
\node at (6,2.5) {$I'$};
\draw[dashed] (6,2.5) ellipse (1.3cm and 0.5cm);
\node at (6,1) {$I$};
\draw[dashed] (6,1) ellipse (1.3cm and 0.5cm);
\node at (6,-0.4) {$3$-$5$-add};
\draw (9,1.75) -- (9,2.15);
\draw (9,1.75) -- (8.2,1);
\draw (9,0) -- (9.5,1);
\draw (9,1.75) -- (9.5,1);
\draw (9,0) -- (10,1);

\draw (9,2.15) -- (8.2,2.9);
\draw (9,2.15) -- (9.5,2.9);
\draw (9,1.75) -- (9.5,2.9);
\draw (9,1.75) -- (10,2.9);
\node at (9,0) [circle,draw=black,fill=black, scale=0.7] {};
\node at (8.7,0) {$u$};
\node at (9,1.75) [circle,draw=red!50,fill=red!20] {};
\node at (8.6,2.15) {$w'_u$};
\node at (9,2.15) [circle,draw=yellow!50,fill=yellow!20] {};
\node at (9,2.15) {$6$};
\node at (9,1.75) {$3$};
\node at (8.6,1.75) {$w_u$};
\node at (8.2,2.9) [circle,draw=blue!50,fill=blue!20] {};
\node at (8.2,2.9) {$1$};
\node at (10,2.9) [circle,draw=blue!50,fill=blue!20] {};
\node at (10,2.9) {$1$};
\node at (9.5,2.9) [circle,draw=blue!50,fill=blue!20] {};
\node at (9.5,2.9) {$1$};
\node at (8.2,1) [circle,draw=blue!50,fill=blue!20] {};
\node at (8.2,1) {$1$};
\node at (10,1) [circle,draw=blue!50,fill=blue!20] {};
\node at (10,1) {$1$};
\node at (9.5,1) [circle,draw=blue!50,fill=blue!20] {};
\node at (9.5,1) {$1$};
\node at (9,2.9) {$I'$};
\draw[dashed] (9,2.9) ellipse (1.3cm and 0.5cm);
\node at (9,1) {$I$};
\draw[dashed] (9,1) ellipse (1.3cm and 0.5cm);
\node at (9,-0.4) {$3$-$6$-add};
\draw (4.5,-5.3) -- (4.5,-4.8);
\draw (4.5,-5.3) -- (5.5,-3.5);
\draw (4.5,-5.3) -- (6.5,-3.5);
\draw (4.5,-4.8) -- (4.5,-4.3);
\draw (4.5,-4.8) -- (2.7,-3.5);
\draw (4.5,-4.8) -- (3.3,-2.4);
\draw (4.5,-4.8) -- (5.3,-2.4);
\draw (4.5,-4.8) -- (6,-2.4);
\draw (4.5,-4.3) -- (3.7,-0.9);
\draw (4.5,-4.3) -- (5,-0.9);
\draw (4.5,-4.3) -- (3.3,-2.4);
\draw (4.5,-1.65) -- (5,-0.9);
\draw (4.5,-1.65) -- (5.5,-0.9);
\node at (3.7,-0.9) [circle,draw=blue!50,fill=blue!20] {};
\node at (3.7,-0.9) {$1$};
\node at (5.5,-0.9) [circle,draw=blue!50,fill=blue!20] {};
\node at (5.5,-0.9) {$1$};
\node at (5,-0.9) [circle,draw=blue!50,fill=blue!20] {};
\node at (5,-0.9) {$1$};
\node at (3.3,-2.4) [circle,draw=blue!50,fill=blue!20] {};
\node at (3.3,-2.4) {$1$};
\node at (5.3,-2.4) [circle,draw=blue!50,fill=blue!20] {};
\node at (5.3,-2.4) {$1$};
\node at (6,-2.4) [circle,draw=blue!50,fill=blue!20] {};
\node at (6,-2.4) {$1$};
\node at (2.7,-3.5) [circle,draw=blue!50,fill=blue!20] {};
\node at (2.7,-3.5) {$1$};
\node at (5.5,-3.5) [circle,draw=blue!50,fill=blue!20] {};
\node at (5.5,-3.5) {$1$};
\node at (6.5,-3.5) [circle,draw=blue!50,fill=blue!20] {};
\node at (6.5,-3.5) {$1$};
\node at (4.5,-5.3) [circle,draw=black,fill=black, scale=0.7] {};
\node at (4.2,-5.3) {$u$};
\node at (4.5,-4.8) [circle,draw=red!50,fill=red!20] {};
\node at (4.5,-4.8) {$2$};
\node at (4.1,-4.8) {$w_u$};
\node at (4.5,-4.3) [circle,draw=green!50,fill=green!20] {};
\node at (4.5,-4.3) {$5$};
\node at (4.1,-4.3) {$w'_u$};
\node at (4.5,-1.65) [circle,draw=red!50,fill=red!20] {};
\node at (4.5,-1.65) {$2$};
\node at (4.1,-1.65) {$w''_u$};
\node at (4.5,-0.9) {$I''$};
\draw[dashed] (4.5,-0.9) ellipse (1.3cm and 0.5cm);
\node at (4.5,-2.4) {$I'$};
\draw[dashed] (4.5,-2.4) ellipse (2.6cm and 0.5cm);
\node at (4.5,-3.5) {$I$};
\draw[dashed] (4.5,-3.5) ellipse (2.6cm and 0.5cm);
\node at (4.5,-5.7) {$2$-$5$-add};
\end{tikzpicture}
\end{center}
\caption{The different used operations in Propositions \ref{bipartit1} and \ref{bipartit2} ($I$, $I'$ and $I''$ are independent sets).}
\label{figajout}
\end{figure}
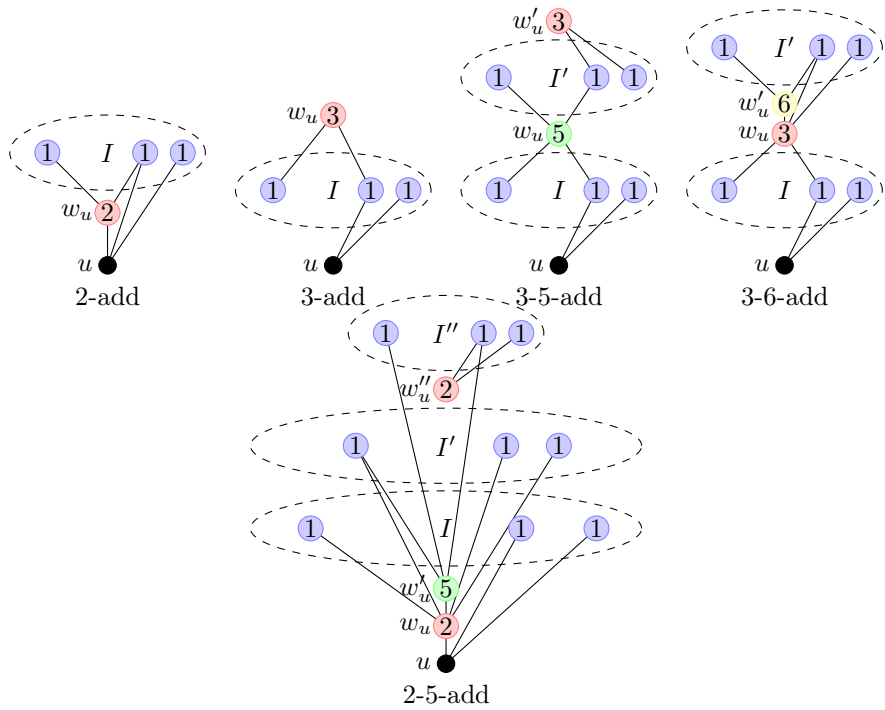
\begin{theo}\label{bipartit1}
Let $S=(1,3,k,k')$, with $5\le k \le 7$ and $6\le k'\le 7$. A graph $G$ is $S$-colorable if and only if one of these three following conditions is true:
\begin{enumerate}
\item $\beta(G)\le 3$;
\item the graph $G$ is obtained from a graph $G'$ with vertex cover $C\subseteq \{u,v\}$ of order at least 2 where we can do the following operations:
\begin{enumerate}
\item if $|C|=2$ and $u$ and $v$ are adjacent, then we can do several $3$-adds on $u$ or $v$ or instead of doing several $3$-adds on a vertex among $u$ and $v$ we can do a $2$-add on this vertex;
\item otherwise, if $|C|<1$ or if $u$ and $v$ are not adjacent, then we can do several $3$-adds on $u$ or $v$ or instead of doing several $3$-adds on $u$ or $v$ we can do a $2$-add on these vertices;
\end{enumerate}
\item the graph $G$ is obtained from a bipartite multi-graph $G'=(A,B)$ in which every edge is subdivided three times and where we can do the following operations:
\begin{enumerate}
\item we can add leaves on the central subdivided vertices and do a $3$-add or adding leaves on vertices that come from $A$ or $B$;
\item if $k\le 6$, we can do a $3$-$6$-add on vertices that come from $A$. If $k'\le6$, we can do a $3$-$6$-add on the vertices that come from $A$ or $B$;
\item if $k=5$, instead of doing a $3$-$6$-add on vertices that come from $A$, we can do a $3$-$5$-add or several $3$-$6$-adds.
\end{enumerate}
\end{enumerate}
\end{theo}
\begin{proof}
Let $G$ be a graph verifying one of the three conditions. If the first condition is true, then we color the vertices of the vertex cover by colors $3$, $k$ and $k'$ and the remaining vertices by color $1$. If the second condition is true, then we color $u$ and $v$ by color $k$ and $k'$, the other vertices that come from $G'$ by color $1$ and the other vertices as in Figure \ref{figajout}. If the third condition is true, then we color the central subdivided vertices by color $3$, the vertices that come from $A$ by color $k'$, the vertices that come from $B$ by color $k$, the vertices that have been added as in Figure \ref{figajout} and the remaining vertices by color $1$. To conclude, if one of the three conditions is true, then $G$ is $S$-colorable.

Suppose that $G$ is $S$-colorable.
If no vertex is colored by $k$ or no vertex is colored by $k'$, then the first condition or the second condition (for a vertex cover of order 1) is true.
Hence, we can suppose that there exists at least one vertex of color $k$ and one vertex of color $k'$.
The proof will be organized as follows: we suppose that there exists a vertex $u$ of color $k$ and we consider a vertex $v$ of color $k'$ at minimal distance from $u$. We will prove that $u$ and $v$ are at distance $4$ or that one vertex among $u$ and $v$ can be obtained by doing one of the previous operation on the other vertex. Since the only way to color the vertices in the path between $u$ and $v$ consists in using a color $3$ on the central vertex (the vertex at distance $2$ from $u$ and $v$) of this path, the third condition is true about $G$ .

If $d(u,v)=1$, then the first or the second condition (case (a)) is true. If $d(u,v)=2$ and if $k>5$, then the first or the second condition (case (b)) is true.
If $d(u,v)=2$ and $k=5$, then $u$ is obtained by $3$-$5$-add on $v$.
If $d(u,v)=3$ and $k=7$, then the first or the second condition (case (b)) is true. If $d(u,v)=3$ and $k'\le 6$ ($k=6$, respectively), then $u$ ($u$ or $v$, respectively) is obtained by $3$-$6$-add on $v$ ($u$ or $v$, respectively).
If $d(u,v)>4$, then the $S$-coloring can not be extended to the path between $u$ and $v$. 
\end{proof}
In the following theorem, by an edge $uv$ coming from a multi-graph $G$, we consider any edge among the multi-edges between $u$ and $v$.
\begin{theo}\label{bipartit2}
Let $S=(1,2,k,k')$, with $5\le k \le 7$ and $6\le k'\le 7$. A graph $G$ is $S$-colorable if and only if one of these three following conditions is true:
\begin{enumerate}
\item $\beta(G)\le 3$;
\item the graph $G$ is obtained from a graph $G'$ with vertex cover $C\subseteq \{u,v\}$ of order at least 2 where we can do several $3$-adds on $u$ or $v$ and doing a $2$-add on $u$ or $v$;
\item the graph $G$ is obtained from a bipartite multi-graph $G'=(A,B)$ where every edge is subdivided three times and where we can do the following operations:
\begin{enumerate}
\item if $k\le 6$ and $k'=7$, we can identify two adjacent subdivided vertices that come from an edge $uv$ where $u$ has only one neighbor in $G'$ and $u$ come from $B$;
\item if $k=6$ and $k'=6$, we can take an independent set of edges from $E(G')$ and identify two adjacent subdivided vertices that come from these edges;
\item if $k=5$ and $k'=6$, we can take a set of edges from $E(G')$ such that no edge have an extremity in common in $B$ and identify two adjacent subdivided vertices that come from these edges;
\item we can add leaves on central subdivided vertices or, if two vertices have been identified, on one of the two subdivided vertices. We can never add leaves on two vertices at distance $2$. We can do $3$-adds, a $2$-add and add leaves on vertices that come from $A$ or $B$. We can not do a $2$-add on a vertex which is adjacent to a vertex coming from an identification;
\item if $k=5$, instead of doing a $2$-add on vertices that come from $A$, we can do a $2$-$5$-add or a $3$-$5$-add on vertices that come from $A$. We can not do a $2$-$5$-add on vertices at distance at most 2 from a vertex coming from an identification.
\end{enumerate}
\end{enumerate}
\end{theo}
\begin{proof}
Let $G$ be a graph verifying one of the three conditions. If the first condition is true, then we color the vertices of the vertex cover by colors $2$, $k$ and $k'$ and the remaining vertices by color $1$. If the second condition is true, then we color $u$ and $v$ by color $k$ and $k'$, the other vertices that come from $G'$ by color $1$ and the other vertices as in Figure \ref{figajout}. If the third condition is true, then we color the central subdivided vertices by color $2$, the vertices that come from $A$ by color $k'$ (if two subdivided vertices have been identified, we color the possible vertex on which leaves have been added by color $2$), the vertices that come from $B$ by color $k$, the vertices that have been added as in Figure \ref{figajout} and the remaining vertices by color $1$. To conclude, if one of the three conditions is true, then $G$ is $S$-colorable.

Suppose that $G$ is $S$-colorable.
If no vertex is colored by $k$ or no vertex is colored by $k'$, then the first condition or the second condition is true.
Hence, we can suppose that there exists at least one vertex of color $k$ and one vertex of color $k'$.
The proof will be organized as follows: we suppose that there exists a vertex $u$ of color $k$ and we consider a vertex $v$ of color $k'$ at minimal distance from $u$. We will prove that $u$ and $v$ are at distance between $3$ and $4$ or prove that one vertex among $u$ and $v$ can be obtained by doing one of the previous operation on the other vertex. Since the only way to color the vertices in the path between $u$ and $v$ consists in using a color $2$ on one of the vertex of this path (the possible vertex of degree at least $3$ or the central vertex if no vertices have been identified), the third condition is true about $G$ .

If $d(u,v)=1$, then the first or the second condition is true. If $d(u,v)=2$ and if $k>5$, then the first or the second condition is true.
If $d(u,v)=2$ and $k=5$, then $u$ is obtained by $2$-$5$-add or by $3$-$5$-add on $v$.
If $d(u,v)>4$, then the $S$-coloring can not be extended to the path between $u$ and $v$. 
\end{proof} 
\begin{figure}
\begin{center}
\begin{tikzpicture}[scale=0.9]
\draw (0+1.5*0.6,0) -- (7*0.6+1.5*0.6,0);
\draw (4*0.6+1.5*0.6,0) -- (4*0.6+1.5*0.6,1);
\node at (0+1.5*0.6,0) [circle,draw=black,fill=black, scale=0.5] {};
\node at (1*0.6+1.5*0.6,0) [circle,draw=black,fill=black, scale=0.5] {};
\node at (2*0.6+1.5*0.6,0) [circle,draw=black,fill=black, scale=0.5] {};
\node at (3*0.6+1.5*0.6,0) [circle,draw=black,fill=black, scale=0.5] {};
\node at (4*0.6+1.5*0.6,0) [circle,draw=black,fill=black, scale=0.5] {};
\node at (5*0.6+1.5*0.6,0) [circle,draw=black,fill=black, scale=0.5] {};
\node at (6*0.6+1.5*0.6,0) [circle,draw=black,fill=black, scale=0.5] {};
\node at (7*0.6+1.5*0.6,0) [circle,draw=black,fill=black, scale=0.5] {};
\node at (4*0.6+1.5*0.6,0.5) [circle,draw=black,fill=black, scale=0.5] {};
\node at (4*0.6+1.5*0.6,1) [circle,draw=black,fill=black, scale=0.5] {};
\node at (0+1.5*0.6,0.4){$\ell'_{1}$};
\node at (7*0.6+1.5*0.6,0.4){$\ell'_{2}$};
\draw (0,-2) -- (10*0.6,-2);
\draw (4*0.6,-2) -- (4*0.6,-1);
\draw (7*0.6,-2) -- (7*0.6,-1);
\node at (0*0.6,-2) [circle,draw=black,fill=black, scale=0.5] {};
\node at (1*0.6,-2) [circle,draw=black,fill=black, scale=0.5] {};
\node at (2*0.6,-2) [circle,draw=black,fill=black, scale=0.5] {};
\node at (3*0.6,-2) [circle,draw=black,fill=black, scale=0.5] {};
\node at (4*0.6,-2) [circle,draw=black,fill=black, scale=0.5] {};
\node at (5*0.6,-2) [circle,draw=black,fill=black, scale=0.5] {};
\node at (6*0.6,-2) [circle,draw=black,fill=black, scale=0.5] {};
\node at (7*0.6,-2) [circle,draw=black,fill=black, scale=0.5] {};
\node at (8*0.6,-2) [circle,draw=black,fill=black, scale=0.5] {};
\node at (9*0.6,-2) [circle,draw=black,fill=black, scale=0.5] {};
\node at (10*0.6,-2) [circle,draw=black,fill=black, scale=0.5] {};
\node at (4*0.6,-1.5) [circle,draw=black,fill=black, scale=0.5] {};
\node at (4*0.6,-1) [circle,draw=black,fill=black, scale=0.5] {};
\node at (7*0.6,-1.5) [circle,draw=black,fill=black, scale=0.5] {};
\node at (7*0.6,-1) [circle,draw=black,fill=black, scale=0.5] {};
\node at (0,-1.6){$\ell'_{1}$};
\node at (10*0.6,-1.6){$\ell'_{2}$};
\draw (0+13.5*0.6,0) -- (7*0.6+13.5*0.6,0);
\draw (4*0.6+13.5*0.6,0) -- (4*0.6+13.5*0.6,1);
\node at (0+13.5*0.6,0) [circle,draw=blue!50,fill=blue!20] {};
\node at (0+13.5*0.6,0) {$3_a$};
\node at (1*0.6+13.5*0.6,0) [circle,draw=red!50,fill=red!20] {};
\node at (1*0.6+13.5*0.6,0) {$3_b$};
\node at (2*0.6+13.5*0.6,0) [circle,draw=green!50,fill=green!20] {};
\node at (2*0.6+13.5*0.6,0) {$2_a$};
\node at (3*0.6+13.5*0.6,0) [circle,draw=yellow!50,fill=yellow!20] {};
\node at (3*0.6+13.5*0.6,0) {$2_b$};
\node at (4*0.6+13.5*0.6,0) [circle,draw=blue!50,fill=blue!20] {};
\node at (4*0.6+13.5*0.6,0) {$3_a$};
\node at (5*0.6+13.5*0.6,0) [circle,draw=green!50,fill=green!20] {};
\node at (5*0.6+13.5*0.6,0) {$2_a$};
\node at (6*0.6+13.5*0.6,0) [circle,draw=yellow!50,fill=yellow!20] {};
\node at (6*0.6+13.5*0.6,0) {$2_b$};
\node at (7*0.6+13.5*0.6,0) [circle,draw=red!50,fill=red!20] {};
\node at (7*0.6+13.5*0.6,0)  {$3_b$};
\node at (4*0.6+13.5*0.6,0.5) [circle,draw=red!50,fill=red!20] {};
\node at (4*0.6+13.5*0.6,0.5) {$3_b$};
\node at (4*0.6+13.5*0.6,1) [circle,draw=green!50,fill=green!20] {};
\node at (4*0.6+13.5*0.6,1) {$2_a$};
\node at (0+13.5*0.6,0.4){$\ell'_{1}$};
\node at (7*0.6+13.5*0.6,0.4){$\ell'_{2}$};
\draw (0+12*0.6,-2) -- (10*0.6+12*0.6,-2);
\draw (4*0.6+12*0.6,-2) -- (4*0.6+12*0.6,-1);
\draw (7*0.6+12*0.6,-2) -- (7*0.6+12*0.6,-1);
\node at (0*0.6+12*0.6,-2) [circle,draw=blue!50,fill=blue!20] {};
\node at (0*0.6+12*0.6,-2) {$3_a$};
\node at (1*0.6+12*0.6,-2) [circle,draw=red!50,fill=red!20] {};
\node at (1*0.6+12*0.6,-2) {$3_b$};
\node at (2*0.6+12*0.6,-2) [circle,draw=green!50,fill=green!20] {};
\node at (2*0.6+12*0.6,-2) {$2_a$};
\node at (3*0.6+12*0.6,-2) [circle,draw=yellow!50,fill=yellow!20] {};
\node at (3*0.6+12*0.6,-2) {$2_b$};
\node at (4*0.6+12*0.6,-2) [circle,draw=blue!50,fill=blue!20] {};
\node at (4*0.6+12*0.6,-2) {$3_a$};
\node at (5*0.6+12*0.6,-2) [circle,draw=green!50,fill=green!20] {};
\node at (5*0.6+12*0.6,-2) {$2_a$};
\node at (6*0.6+12*0.6,-2) [circle,draw=yellow!50,fill=yellow!20] {};
\node at (6*0.6+12*0.6,-2) {$2_b$};
\node at (7*0.6+12*0.6,-2) [circle,draw=red!50,fill=red!20] {};
\node at (7*0.6+12*0.6,-2)  {$3_b$};
\node at (4*0.6+12*0.6,-1.5) [circle,draw=red!50,fill=red!20] {};
\node at (4*0.6+12*0.6,-1.5) {$3_b$};
\node at (4*0.6+12*0.6,-1) [circle,draw=green!50,fill=green!20] {};
\node at (4*0.6+12*0.6,-1) {$2_a$};
\node at (8*0.6+12*0.6,-2) [circle,draw=green!50,fill=green!20] {};
\node at (8*0.6+12*0.6,-2) {$2_a$};
\node at (9*0.6+12*0.6,-2) [circle,draw=yellow!50,fill=yellow!20] {};
\node at (9*0.6+12*0.6,-2) {$2_b$};
\node at (10*0.6+12*0.6,-2) [circle,draw=blue!50,fill=blue!20] {};
\node at (10*0.6+12*0.6,-2) {$3_a$};
\node at (7*0.6+12*0.6,-1.5) [circle,draw=blue!50,fill=blue!20] {};
\node at (7*0.6+12*0.6,-1.5) {$3_a$};
\node at (7*0.6+12*0.6,-1) [circle,draw=green!50,fill=green!20] {};
\node at (7*0.6+12*0.6,-1) {$2_a$};
\node at (0+11*0.6,-1.6){$\ell'_{1}$};
\node at (10*0.6+11*0.6,-1.6){$\ell'_{2}$};
\end{tikzpicture}
\end{center}
\caption{The graphs $L'$ (at the bottom) and $J'$ (at the top) and a $(2,2,3,3)$-coloring of $L'$ and $J'$ (on the right).}
\label{mJK}
\end{figure}
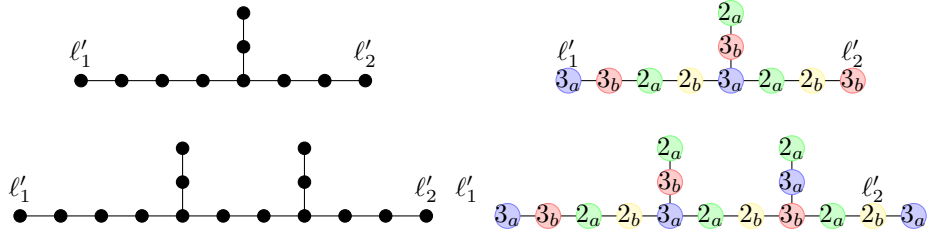
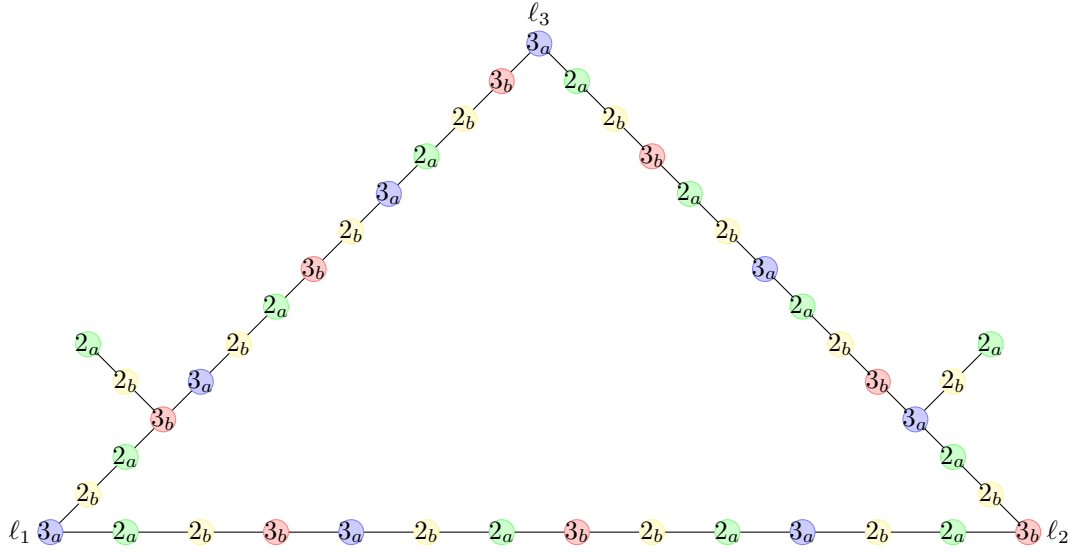
\begin{figure}[t]
\begin{center}
\begin{tikzpicture}
\draw (0,0) -- (13,0);
\draw (0,0) -- (6.5,6.5);
\draw (6.5,6.5) -- (13,0);
\draw (1.5,1.5) -- (0.5,2.5);
 \draw (11.5,1.5) -- (12.5,2.5);
\node at (0,0) [circle,draw=blue!50,fill=blue!20] {};
\node at (0,0) {$3_a$};
\node at (1*0.5,1*0.5) [circle,draw=yellow!50,fill=yellow!20] {};
\node at (1*0.5,1*0.5) {$2_b$};
\node at (2*0.5,2*0.5) [circle,draw=green!50,fill=green!20] {};
\node at (2*0.5,2*0.5) {$2_a$};
\node at (3*0.5,3*0.5) [circle,draw=red!50,fill=red!20] {};
\node at (3*0.5,3*0.5) {$3_b$};
\node at (4*0.5,4*0.5) [circle,draw=blue!50,fill=blue!20] {};
\node at (4*0.5,4*0.5) {$3_a$};
\node at (5*0.5,5*0.5) [circle,draw=yellow!50,fill=yellow!20] {};
\node at (5*0.5,5*0.5) {$2_b$};
\node at (6*0.5,6*0.5) [circle,draw=green!50,fill=green!20] {};
\node at (6*0.5,6*0.5) {$2_a$};
\node at (7*0.5,7*0.5) [circle,draw=red!50,fill=red!20] {};
\node at (7*0.5,7*0.5) {$3_b$};
\node at (8*0.5,8*0.5) [circle,draw=yellow!50,fill=yellow!20] {};
\node at (8*0.5,8*0.5) {$2_b$};
\node at (9*0.5,9*0.5) [circle,draw=blue!50,fill=blue!20] {};
\node at (9*0.5,9*0.5) {$3_a$};
\node at (10*0.5,10*0.5) [circle,draw=green!50,fill=green!20] {};
\node at (10*0.5,10*0.5) {$2_a$};
\node at (11*0.5,11*0.5) [circle,draw=yellow!50,fill=yellow!20] {};
\node at (11*0.5,11*0.5) {$2_b$};
\node at (12*0.5,12*0.5) [circle,draw=red!50,fill=red!20] {};
\node at (12*0.5,12*0.5) {$3_b$};
\node at (6.5,6.5) [circle,draw=blue!50,fill=blue!20] {};
\node at (6.5,6.5) {$3_a$};
\node at (13,0) [circle,draw=red!50,fill=red!20] {};
\node at (13,0) {$3_b$};
\node at (13-1*0.5,1*0.5) [circle,draw=yellow!50,fill=yellow!20] {};
\node at (13-1*0.5,1*0.5) {$2_b$};
\node at (13-2*0.5,2*0.5) [circle,draw=green!50,fill=green!20] {};
\node at (13-2*0.5,2*0.5) {$2_a$};
\node at (13-4*0.5,4*0.5) [circle,draw=red!50,fill=red!20] {};
\node at (13-4*0.5,4*0.5) {$3_b$};
\node at (13-3*0.5,3*0.5) [circle,draw=blue!50,fill=blue!20] {};
\node at (13-3*0.5,3*0.5) {$3_a$};
\node at (13-5*0.5,5*0.5) [circle,draw=yellow!50,fill=yellow!20] {};
\node at (13-5*0.5,5*0.5) {$2_b$};
\node at (13-6*0.5,6*0.5) [circle,draw=green!50,fill=green!20] {};
\node at (13-6*0.5,6*0.5) {$2_a$};
\node at (13-7*0.5,7*0.5) [circle,draw=blue!50,fill=blue!20] {};
\node at (13-7*0.5,7*0.5) {$3_a$};
\node at (13-8*0.5,8*0.5) [circle,draw=yellow!50,fill=yellow!20] {};
\node at (13-8*0.5,8*0.5) {$2_b$};
\node at (13-9*0.5,9*0.5) [circle,draw=green!50,fill=green!20] {};
\node at (13-9*0.5,9*0.5) {$2_a$};
\node at (13-10*0.5,10*0.5) [circle,draw=red!50,fill=red!20] {};
\node at (13-10*0.5,10*0.5) {$3_b$};
\node at (13-11*0.5,11*0.5) [circle,draw=yellow!50,fill=yellow!20] {};
\node at (13-11*0.5,11*0.5) {$2_b$};
\node at (13-12*0.5,12*0.5) [circle,draw=green!50,fill=green!20] {};
\node at (13-12*0.5,12*0.5) {$2_a$};
\node at (1,0) [circle,draw=green!50,fill=green!20] {};
\node at (1,0) {$2_a$};
\node at (2,0) [circle,draw=yellow!50,fill=yellow!20] {};
\node at (2,0) {$2_b$};
\node at (3,0) [circle,draw=red!50,fill=red!20] {};
\node at (3,0) {$3_b$};
\node at (4,0) [circle,draw=blue!50,fill=blue!20] {};
\node at (4,0) {$3_a$};
\node at (6,0) [circle,draw=green!50,fill=green!20] {};
\node at (6,0) {$2_a$};
\node at (5,0) [circle,draw=yellow!50,fill=yellow!20] {};
\node at (5,0) {$2_b$};
\node at (7,0) [circle,draw=red!50,fill=red!20] {};
\node at (7,0) {$3_b$};
\node at (9,0) [circle,draw=green!50,fill=green!20] {};
\node at (9,0) {$2_a$};
\node at (8,0) [circle,draw=yellow!50,fill=yellow!20] {};
\node at (8,0) {$2_b$};
\node at (10,0) [circle,draw=blue!50,fill=blue!20] {};
\node at (10,0) {$3_a$};
\node at (12,0) [circle,draw=green!50,fill=green!20] {};
\node at (12,0) {$2_a$};
\node at (11,0) [circle,draw=yellow!50,fill=yellow!20] {};
\node at (11,0) {$2_b$};

\node at (1,2) [circle,draw=yellow!50,fill=yellow!20] {};
\node at (1,2) {$2_b$};
\node at (0.5,2.5) [circle,draw=green!50,fill=green!20] {};
\node at (0.5,2.5) {$2_a$};
\node at (12,2) [circle,draw=yellow!50,fill=yellow!20] {};
\node at (12,2) {$2_b$};
\node at (12.5,2.5) [circle,draw=green!50,fill=green!20] {};
\node at (12.5,2.5) {$2_a$};
\node at (-0.4,0) {$\ell_{1}$};
\node at (13.4,0) {$\ell_{2}$};
\node at (6.5,6.9) {$\ell_{3}$};
\end{tikzpicture}
\end{center}
\caption{The graph $H'$ and a $(2,2,3,3)$-coloring of it.}
\label{mH}
\end{figure}
In the following proof, we use the construction from Lemma \ref{reffinal}. The existence properties of an $S$-coloring with the properties of Lemma \ref{reffinal} will not be verified anymore. However, using the different $(2,2,3,3)$-colorings from the proof, the reader can be convinced that there exists an $S$-coloring of the whole constructed graph.
\begin{theo}
The problem $(2,2,3,3)$-COL is NP-complete.
\end{theo}

\begin{proof}
Let $L'$ and $J'$ be subcubic graphs from Figure~\ref{mJK} with connectors the vertices $\ell'_1$ and $\ell'_2$. Note that a vertex of degree at least $3$ can only be colored by colors $3_a$ or $3_b$. In the constructed graph (Lemma \ref{reffinal}), a vertex of degree at least $3$ will be in an induced $Y_2$ and there does not exist an $S$-coloring of $Y_2$ with the central vertex colored by colors $2_a$ or $2_b$. For the same reason, every connectors should be colored by colors $3_a$ or $3_b$ (these vertices will be in an induced $Y_2$ in the constructed graphs). If two vertices $u$ and $v$ have degree at least $3$ or are connectors and $d(u,v)=4$, then $u$ and $v$ can only have the same color $3$, since the vertices in the path between $u$ and $v$ can not be colored anymore if $u$ and $v$ have different colors.

If two connectors of $L'$ have the same color, then the vertices in the path of length $4$ between the central vertex of degree 3 and a connector can not be colored. Thus, in every $(2,2,3,3)$-coloring of $L'$, the connectors should have the same color $3$. A $(2,2,3,3)$-coloring of $L'$ is illustrated in Figure~\ref{mJK}.

If two connectors of $J'$ have different colors $3$, then the vertices in the path of length $4$ between one of the central vertex of degree 3 and a connector can not be colored. Thus, in every $(2,2,3,3)$-coloring of $J'$, the connectors can not have the same color $3$. A $(2,2,3,3)$-coloration of $J'$ is illustrated in Figure~\ref{mJK}. Hence, $L'$ is a $3$-$3$-transmitter and $J'$ is a $3$-$3$-antitransmitter.

Let $H'$ be the subcubic graph from Figure~\ref{mH} with connectors the vertices $\ell_1$, $\ell_2$ and $\ell_3$. The graph $H'$ contains three connectors and two vertices of degree $3$. These vertices should be all colored by colors $3_a$ or $3_b$. Suppose that $\ell_1$, $\ell_2$ and $\ell_3$ have the same color $3_a$. Let $u$ be one of the vertex at distance $3$ from $\ell_1$.
Suppose that $u$ is colored by color $3_b$. The coloring can not be extended to the vertices in the path between $u$ and the vertex of degree $3$ at minimum distance of $u$, since no vertices in this path can be colored by $3_b$ and $P_6$ is not $(2,2,3)$-colorable. Thus, the vertices at distance $3$ from $\ell_1$ can not be colored by $3_b$. One vertex among $u$ and the vertices in the shortest path between $u$ and $\ell_1$ must be colored by $3_b$. There does not exist a $(2,2,3,3)$-coloring of $Y_2$ for which a vertex at distance $2$ from a vertex of degree $3$ is colored by a color $3$. Thus, the only way to color $H$ is to color the neighbors of $\ell_3$ by color $3_b$.
Hence, we have a contradiction (the distance between these two vertices of color $3_b$ is only $2$). Consequently, in every $(2,2,3,3)$-coloring of $H'$, the connectors can not have a same color.

Figure~\ref{mH} illustrates a $(2,2,3,3)$-coloring for $\ell_1$ and $\ell_3$ with the same color (in the other cases, the $(2,2,3,3)$-colorings are similar). To conclude, $H'$ is a $3$-$3$-clause simulator.
\end{proof}
\begin{proof}[Proof of Theorem \ref{theo31}]
By Propositions \ref{dichoa}, \ref{dicho3}, \ref{dicho4}, \ref{bipartit1} and \ref{bipartit2}, we obtain that $S$-COL is polynomial-time solvable in the given cases.
By Propositions \ref{dicho0}, \ref{dicho1} and \ref{dicho2}, we obtain that $S$-COL is NP-complete in the other cases.
\end{proof}

\section{Conclusion}
We have proven several dichotomy theorems about $S$-COL.
For several instances of $S$-COL ($(2,3,3,3)$-COL, for example) the family of $S$-colorable graphs is included in some known families of graphs, since every $S$-colorable graph has bounded pathwidth.
For other instances of $S$-COL ($(2,2,3,3)$-COL, for example), $S$-COL is NP-complete, by giving a reduction.
For a last kind of instances of $S$-COL ($(1,2,7,7)$-COL, for example), we determine a construction of the $S$-colorable graphs. The structure of these $S$-colorable graphs is not so far from the structure of a bipartite graph.

We finish by proposing the following conjecture about $S$-COL.
\begin{con}
There exists a dichotomy between NP-complete problems and polynomial-time solvable problems among the instances of $S$-COL.
\end{con}

\end{document}